\documentclass[copyright,creativecommons]{eptcs}
\providecommand{\event}{SOS 2007} % Name of the event you are submitting to
\providecommand{\publicationstatus}{Supplementary Report}
\usepackage{underscore}           % Only needed if you use pdflatex.

\usepackage[utf8]{inputenc}

\usepackage[group-separator={,},binary-units]{siunitx}
\usepackage{paralist}
\usepackage{booktabs}
\usepackage{subcaption}
\usepackage{listings}
\usepackage[dvipsnames,table,xcdraw]{xcolor}
\usepackage{tikz}
\usetikzlibrary{arrows,positioning,petri,decorations,shapes}

\tikzset{
  every place/.style=
    {
      circle,
      draw,
      thick,
      inner sep=3pt,
      minimum size=7mm
    },
  every transition/.style=
    {
      rectangle,
      draw,
      thick,
      inner sep=3pt,
      minimum size=7mm
    },
  edge/.style=
    {
      ->,
      shorten <=1pt,
      >=stealth',
      semithick
    },
  pre/.style=
    {
      <-,shorten <=.1pt,>=stealth',semithick
    },
  post/.style=
    {
      ->,shorten >=.1pt,>=stealth',semithick
    },
  state/.style=
    {
      circle,draw,semithick,inner sep=.1pt,minimum size=1.5mm,fill=black
    },
  entity/.style=
    {
      rounded corners=3,
      draw,
      semithick,
      inner sep=.5em,
      minimum size=2em
    }
}

\usepackage[linesnumbered,ruled]{algorithm2e}
\usepackage{amsmath,amssymb,amsfonts,stmaryrd,ifsym}
\usepackage[thmmarks,thref,amsmath,framed]{ntheorem}
\theoremstyle{plain}
\theoremheaderfont{\normalfont\bfseries}
\theorembodyfont{\itshape}
\newtheorem{theorem}{Theorem}
\newtheorem{lemma}{Lemma}
\newtheorem{proposition}{Proposition}

\theoremstyle{plain}
\theoremheaderfont{\bfseries}
\theorembodyfont{\upshape}
\theoremsymbol{\ensuremath{_\blacksquare}}
\newtheorem{definition}{Definition}

\theoremstyle{nonumberplain}
\theoremheaderfont{\scshape}
\theorembodyfont{\upshape}
\theoremseparator{:}
\theoremsymbol{\ensuremath{_\Box}}
\newtheorem{proof}{Proof}

\usepackage{commands,macros,graphs,mycommands}

\title{Fast Dual Simulation Processing of Graph Database Queries\\ (Supplement)}
\author{Stephan Mennicke \qquad Jan-Christoph Kalo \\ Denis Nagel \qquad Hermann Kroll \qquad Wolf-Tilo Balke
\institute{Institut für Informationssysteme,
TU Braunschweig,
Braunschweig, Germany}
\email{\quad \{mennicke,kalo,kroll,balke\}@ifis.cs.tu-bs.de \qquad\qquad denis.nagel@tu-bs.de}
}
\def\titlerunning{Computing Dual Simulation for Graph Database Queries}
\def\authorrunning{S.\ Mennicke et al.}

\usepackage{comment}
\excludecomment{conference}
\includecomment{report}

\newcommand\conf[1]{}
\begin{conference}
\renewcommand\conf[1]{#1}
\end{conference}

\newcommand\rep[1]{}
\begin{report}
\renewcommand\rep[1]{#1}
\end{report}

\newcommand\np{\textsc{np}\xspace}

\begin{document}
\maketitle

\hyphenation{sparqlSim}

\section{Introduction}\label{sec:introduction}
%
% -*- root: ../main.tex -*-
%
\begin{figure*}[t]
\centering
\begin{subfigure}[b]{.75\textwidth}
\resizebox{.95\textwidth}{!}{
\begin{tikzpicture}[node distance=2.5em and 7em]
\node[entity,draw=white] (h1) {};
\node[entity, below=of h1, line width=3pt] (f1) {\tt Mission: Impossible};
\node[entity, below=1.1cm of f1] (a1) {\tt Oscar}
  edge[pre] node[left] {\tt awarded} (f1);
\node[entity, above=1.1cm of f1, line width=3pt] (d1) {\tt B. De Palma}
  edge[post, line width=3pt] node[left] {\tt directed} (f1);
\node[entity, above=1.1cm of d1] (c1) {\tt Newark}
  edge[pre] node[right] {\tt born\_in} (d1);

\node[entity, right=of f1] (t1) {\tt Action}
  edge[pre] node[below] {\tt genre} (f1);

\node[entity, right=of t1, line width=3pt] (f2) {\tt Goldfinger}
  edge[post] node[below] {\tt genre} (t1);
\node[entity, above=1.1cm of f2, line width=3pt] (d2) {\tt G. Hamilton}
  edge[post, line width=3pt] node[left] {\tt directed} (f2);
\node[entity, above=1.1cm of d2] (c2) {\tt Paris}
  edge[pre] node[left] {\tt born\_in} (d2);

\node[entity, below=1.1cm of t1] (f3) {\tt Thunderball}
  edge[post] node[above] {\tt awarded} (a1)
  edge[post] node[auto,swap] {\tt sequel\_of} (f2);

\node[entity, right=of d2, line width=3pt] (d3) {\tt H. Saltzman}
  edge[pre, line width=3pt] node[above] {\tt worked\_with} (d2);
\node[entity, above=1.1cm of d3] (c3) {\tt Saint John}
  edge[pre] node[right] {\tt born\_in} (d3);
\node[entity, right=of f2, draw=white] (h2) {};
\node[entity, right=of h2, line width=1.5pt] (f4) {\tt From Russia with Love}
  edge[post] node[below] {\tt prequel\_of} (f2);
\node[entity, above=1.1cm of f4, line width=1.5pt] (d4) {\tt T. Young}
  edge[post, line width=1.5pt] node[right] {\tt directed} (f4);
\node[entity, below=1.1cm of f4] (a2) {\tt BAFTA Awards}
  edge[pre] node[right] {\tt awarded} (f4);

\node[entity, left=of d1, line width=3pt] (d5) {\tt D. Koepp}
  edge[pre, line width=3pt] node[above] {\tt worked\_with} (d1);
\node[entity, below=1.1cm of d5, line width=1.5pt] (f5) {\tt Mortdecai}
  edge[pre, line width=1.5pt] node[left] {\tt directed} (d5);

\node[entity, above=1.1cm of d5] (p1) {\tt 277.140}
  edge[pre] node[above] {\tt population} (c1);
\node[entity, above=1.1cm of t1, draw=white] (h2) {};
\node[entity, above=1.1cm of h2] (p2) {\tt 2.220.445}
  edge[pre] node[above] {\tt population} (c2);
\node[entity, above=1.1cm of d4] (p3) {\tt 70.063}
  edge[pre] node[above] {\tt population} (c3);

\node[entity, above=1.1cm of t1, dashed] (d6) {\tt P.R. Hunt}
  edge[post, dashed] node[above] {\tt worked\_with} (d1);
\end{tikzpicture}
}
\caption{}
\end{subfigure}
\begin{subfigure}[b]{.14\textwidth}
\begin{center}
\resizebox{.9\textwidth}{!}{
    \begin{tikzpicture}
    \node[entity] (p1) {\tt director};
    \node[entity,below=of p1] (m) {\tt movie}
      edge[pre] node[auto]{\tt directed} (p1);
    \node[entity,above=of p1] (p2) {\tt coworker}
      edge[pre] node[auto]{\tt worked\_with} (p1);
  \end{tikzpicture}
}
\end{center}
\caption{}
\end{subfigure}
\caption{Representation of (a) an Example Graph Database and (b) a Graph Pattern for \X[1]}
\label{fig:database}
\end{figure*}
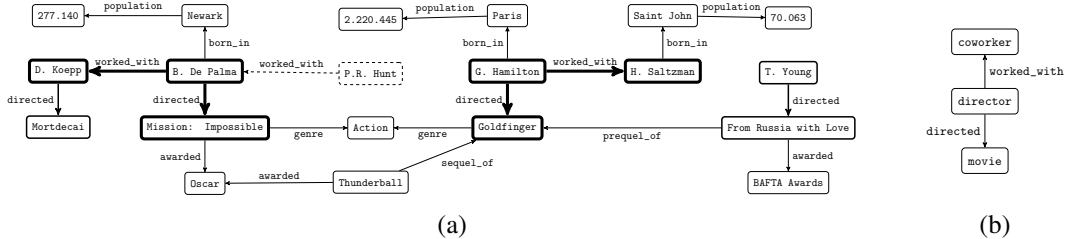
Extensive knowledge graphs are commonplace backbones in today's information infrastructures.
Therefore, scalable query processing in graph databases has sparked a vivid interest in the database community.
Already at an early stage specialized graph query languages such as \sparql, the W3C recommendation for querying RDF data by SQL-like expressions~\cite{Prudhommeaux2008}, have been designed.
Such languages provide easy to use yet expressive query capabilities on graph structures, but need to severely break down structural complexity to allow for fast query evaluation.
Indeed, the evaluation of complex graph patterns is computationally expensive and thus a variety of implementational avenues have been proposed~\cite{Erling2009RDFDBMS,Nenov2015,Atre2010}.

At the heart of \sparql, basic graph patterns (BGPs) form the syntactically least complex queries.
BGPs are simply graphs, and their result sets contain all graph-homomorphic matches from the graph database instance.
Consider query \X[1], retrieving all persons (\cf variable {\tt ?director}) who directed at least one movie ({\tt ?movie}) and at some point collaborated with another person ({\tt ?coworker}):
\begin{center}
\begin{conference}
\footnotesize
\mbox{\begin{tabular}{l}
{\tt SELECT $^*$ WHERE \{} \\
{\tt ~~?director directed ?movie .} \\
{\tt ~~?director worked\_with ?coworker .}
{\tt \}}
\end{tabular}}
\hspace{2em}\X[1]
\end{conference}
\begin{report}
\mbox{\begin{tabular}{l}
{\tt SELECT $^*$ WHERE \{} \\
{\tt ~~?director directed ?movie .} \\
{\tt ~~?director worked\_with ?coworker .}
{\tt \}}
\end{tabular}}
\hspace{2em}\X[1]
\end{report}
\end{center}

\X[1] consists of two triple patterns.
The first requires a {\tt directed} link between assignments to variables {\tt ?director} and {\tt ?movie} while the second asks for {\tt ?director} to be in a {\tt worked\_with} relationship with an object matching {\tt ?coworker}.
An evaluation of \X[1] \wrt the database instance depicted in Fig.~\ref{fig:database}(a) retrieves the two subgraphs in bold print, including nodes {\tt B. De Palma} or {\tt G. Hamilton} assigned to variable {\tt ?director}.

Besides full-fledged graph query languages simpler graph pattern matching for diverse querying tasks raised a growing interest in the database community~\cite{Brynielsson2010,Fan2010,Fan2010Simulation,Fan2012,Lee2013,Ma2014,Fan2015,Mottin2016,Xie2017PoPanda}.
Some of these applications employ a form of {\em simulation graph pattern matching}, showing computational advantages over homomorphic and isomorphic matching.
Yet, an in-depth analysis of the approaches incorporating simulation~\cite{Brynielsson2010,Ma2014,Mottin2016,Xie2017PoPanda} reveals two shortcomings:
\begin{enumerate}[(1)]
\item The algorithms presented are {\em not specifically designed for graph database querying} tasks, in contrast to state-of-the-art graph database management systems like Virtuoso~\cite{Erling2009RDFDBMS}.
Thus, when it comes to performance evaluation of the simulation algorithms, they are only compared to subgraph isomorphism algorithms.
But as their claimed application area is indeed database querying, it would only be fair to test these algorithms against established database systems, too (note that all isomorphism queries can be easily translated to \sparql queries with conjunction and filter conditions~\cite{Mennicke2017LWDA}).
While performance evaluations of graph pattern matching papers generally show good evaluation times, based on our experience we have reason to believe that Virtuoso and other graph database systems would still perform much better.
Therefore, we have to find out whether we can algorithmically catch up with graph database systems, since general simulation queries may not be easily expressed in \sparql~\cite{Mennicke2017LWDA}.
\item What all the classical graph pattern matching problems have in common, is that {\em the input is given as a graph}, \ie there is no possibility of building more complex patterns as by graph query languages.
Hence, we have to study whether there are major boundaries for an incorporation of graph query operators into the pattern matching process.
\end{enumerate}
Towards (1) we investigate dual simulation, a version of simulation specifically developed for the graph data setting~\cite{Ma2014}.
The algorithm presented by Ma \etal follows a single passive strategy that checks whether the definition of dual simulation is met resulting in a huge amount of iterations and influencing the overall runtime (\cf Table~\ref{tab:simulation}).
Based on a novel characterization of dual simulation in Sect.~\ref{sec:soi}, we develop a more flexible algorithmic solution to the dual simulation problem: the fixpoint of a system of inequalities (SOI) allows for {\em fast dual simulation processing} in the graph query setting.
We provide formal proof of the correctness of our algorithm as well as experimental justification for the performance improvements brought by our solution.
\begin{report}
And what is more, our algorithm is also applicable to highly compressed database formats, as \eg the {\em BitMat} storage structure~\cite{Atre2015LeftDescriptors}, and to massive parallelization techniques of bit-matrix operations.
\end{report}

Regarding (2), we also contribute a conservative extension of dual simulation to work with typical graph query operators, exemplarily taken from \sparql (\cf Sect.~\ref{sec:pruning}).
We obtain an overapproximation of the actual \sparql query results for further inspection, filtering, or actual query processing, depending on the specific application.
These extensions are {\it complete} in that none of the matches under the \sparql semantics is neglected by dual simulation.
In particular, this allows for {\it sound} pruning and in any case makes it safe to use the result for further query processing.
Our algorithmic framework remains efficient, since all the features we need to add are directly implementable within the SOI solution and do not influence the overall polynomial-time complexity.
\begin{report}
We do not only deal with well-designed patterns~\cite{Perez2009,Atre2015LeftDescriptors}.
Although well-designed patterns have been of special interest, recent studies show that non-well-designed patterns cannot be neglected, since they form a sizeable portion of practical query loads~\cite{Han2016OnQueries}.
Therefore, we may expect usage of non-well-designed patterns in the above-mentioned applications.
The advantage of our extended dual simulation process is that is does not need to tell non-well-designed patterns apart from well-designed ones.
\end{report}

In Sect.~\ref{sec:evaluation}, we perform extensive experiments on two large-scale databases.
First, we provide evidence of the runtime improvements over the algorithm by Ma \etal due to our solution.
Second, we step into one possible application, namely {\em per-query database pruning}.
More than 95\% irrelevant triples are disqualified by dual simulation processing for all evaluated queries, which is the reason for improved query evaluation times compared to two state-of-the-art graph databases Virtuoso~\cite{Erling2009RDFDBMS} and RDFox~\cite{Nenov2015}.
Moreover, we observe that our dual simulation process may directly be incorporated as a pruning preprocessing step in RDFox.
In Sect.~\ref{sec:related-work}, we elaborate on related work while we draw a conclusion in Sect.~\ref{sec:conclusion}.
\begin{conference}
Most proofs of the formal results had to be omitted due to space limitations.
However, they can be found in the supplementary report~\cite{Mennicke2019TR} accompanying this paper.
\end{conference}

\section{Graphs, Data and Matching}\label{sec:preliminaries}
%
% -*- root: ../main.tex -*-
%
By {\em graphs} we refer to edge-labeled directed graphs with a finite set of nodes $V$, a finite label alphabet $\Sigma$, and a directed labeled edge relation $E\subseteq V\times \Sigma\times V$.
A {\em graph} is a triple $G=(V,\Sigma,E)$ of the aforementioned components.
As exemplified in Fig.~\ref{fig:database}, nodes are depicted as rounded-corner rectangles (with its identifier/name as centered label) while edges are represented by directed arrows (with associated labels next to the arrow) between nodes.
We often identify the components of of graphs $G_i$ by $V_i$ and $E_i$ ($i\in\mathbb N$).
As a matter of simplicity we assume all graphs to be labeled over a fixed alphabet $\Sigma$.
For every label $a\in\Sigma$, we associate with graphs $G$ two adjacency maps, a forward map $\mathfrak{F}_G^a$ and a backward map $\mathfrak{B}_G^a$ of $G$.
Both mappings associate a subset of nodes with each node $v\in V$, in case of forward maps, the set of successor nodes, and in case of backward maps, the set of predecessor nodes of $v$, \ie
% $$\begin{array}{rcll}
% \mathfrak{F}_G^a(v) & := & \{ w \mid (v, a, w) \in E \} & \text{and}  \\
$\mathfrak{F}_G^a(v) := \{ w \mid (v, a, w) \in E \}$ and
% \mathfrak{B}_G^a(v) & := & \{ u \mid (u, a, v) \in E \}\text.
$\mathfrak{B}_G^a(v) := \{ u \mid (u, a, v) \in E \}$.
% \end{array}$$

\begin{report}
In the {\em Resource Description Framework} (RDF), the basic ingredients are {\em triples} $(s,p,o)$, describing a relationship ($p$) between two database resources $s$ and $o$.
By analogy, $s$, $p$ and $o$ are thought of as {\em subject}, {\em predicate} and {\em object}.
Database resources ($s$ or $o$) stem from two universes: the set of all objects $\mathcal O$, each of which usually referred to by an IRI (Internationalized Resource Identifier), and the set of literals $\mathcal L$.
A literal is an element from an arbitrary data domain, such as the integers, usually to describe attribute values of objects.
Predicates are also implemented by IRIs, which stem from the universe $\mathcal P$.
To simplify the presentation we assume all three universes to be disjoint.
Furthermore, we abstract from the implementation as IRIs and use intuitive names to identify database objects and predicates (\cf example database in Fig.~\ref{fig:database}(a)).
RDF allows for generalized triples of type $\mathcal O \times \mathcal P \times ( \mathcal O \cup \mathcal L )$, sufficient to formulate interrelations and attributes of objects.
Attributes connect objects with literals, \eg in Fig.~\ref{fig:database}(a), the information that \texttt{Saint John} has $\texttt{70,063}$ inhabitants is reflected by the triple $(\texttt{Saint John}, \texttt{population}, \texttt{70,063})$.
Further note that literals may only occur in the third component of a triple.
\end{report}

\begin{report}
A {\em graph database} is a finite instance of all possible triples.
We formalize it as a graph with all objects and literals occurring in triples as the set of nodes, and all predicates as the alphabet.
Handling literals properly leads to the following divergence from our initial graph model.
\begin{definition}[Graph Database]\label{def:graphdb}
A {\em graph database} is a graph $\db = ( O_\db, \Sigma, E_\db )$ with a finite set of database objects and literals $O_\db \subset_{\textit{fin}} \mathcal O\cup \mathcal L$, a finite set of properties $\Sigma \subset_{\textit{fin}} \mathcal P$, and a labeled edge relation $E_\db \subseteq ( O_\db \cap \mathcal O) \times \Sigma \times O_\db$.
\end{definition}
All the notions for graphs carry over to graph databases.
\end{report}

\begin{conference}
Graph databases are graphs which associate database objects, \eg entities and literals, with each other via predicates.
To distinguish graph databases from ordinary graphs we denote them by $\db=(O_\db,\Sigma,E_\db)$, where $O_\db$ is the set of database objects and elements of $E_\db$ are sometimes called links.
We refrain from making the data model more concrete, since all upcoming notions and techniques are independent of any further restrictions, as \eg given by RDF's requirement of having literals only as edge targets.
\end{conference}

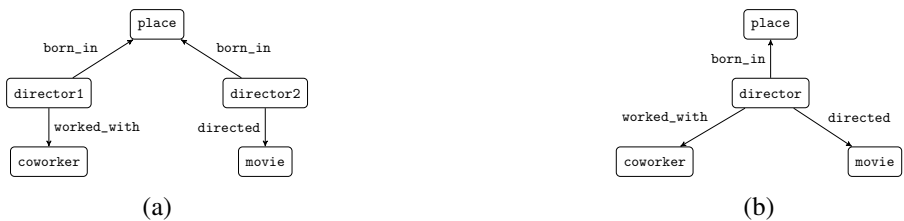
\begin{figure}[bp]
\centering
\begin{subfigure}[b]{.47\linewidth}
  \centering
  \scalebox{.51}{
    \begin{tikzpicture}
      \node[entity] (place) {\texttt{place}};

      \node[entity, below left=of place] (director1) {\texttt{director1}}
        edge[post] node[auto]{\texttt{born\_in}} (place);
      \node[entity, below right=of place] (director2) {\texttt{director2}}
        edge[post] node[auto,swap]{\texttt{born\_in}} (place);

      \node[entity,below=of director1] (coworker) {\texttt{coworker}}
        edge[pre] node[auto,swap]{\texttt{worked\_with}} (director1);

      \node[entity,below=of director2] (movie) {\texttt{movie}}
        edge[pre] node[auto]{\texttt{directed}} (director2);
    \end{tikzpicture}
  }
  \caption{}
\end{subfigure}
\quad
\begin{subfigure}[b]{.47\linewidth}
  \centering
  \scalebox{.51}{
    \begin{tikzpicture}
      \node[entity] (place) {\texttt{place}};

      \node[entity, below=of place] (director1) {\texttt{director}}
        edge[post] node[auto]{\texttt{born\_in}} (place);

      \node[entity,below left=of director1] (coworker) {\texttt{coworker}}
        edge[pre] node[auto]{\texttt{worked\_with}} (director1);

      \node[entity,below right=of director1] (movie) {\texttt{movie}}
        edge[pre] node[auto,swap]{\texttt{directed}} (director1);
    \end{tikzpicture}
  }
  \caption{}
\end{subfigure}
  \caption{Two Graph Patterns}
  \label{fig:dualsim-patterns}
\end{figure}
A dual simulation~\cite{Ma2014} between two graphs $G_1, G_2$ is a binary relation $S\subseteq V_1 \times V_2$ such that for each pair of nodes $(v_1, v_2)\in S$, all incoming and outgoing edges of $v_1$ are also featured by $v_2$ and the adjacent nodes of $v_1$ and $v_2$, ordered in pairs, belong to $S$.
For a dual simulation $S$, $(v_1,v_2)\in S$ means that $v_2$ dual simulates $v_1$.
As an example consider the graphs depicted in Fig.~\ref{fig:dualsim-patterns}(a) and (b) as $G_1$ and $G_2$.
A dual simulation relates the nodes with the same label, \eg \texttt{place} in $G_2$ dual simulates node \texttt{place} in $G_1$, and both nodes \texttt{director1} and \texttt{director2} in $G_1$ relate to \texttt{director} in $G_2$, as in
\begin{equation}\label{eq:dualsim-small-example}
{\fontsize{7.7}{7.7}\selectfont
  \left\{ \begin{array}{l}
    (\texttt{place}, \texttt{place}), (\texttt{director1},\texttt{director}), \\
    (\texttt{director2},\texttt{director}), (\texttt{movie},\texttt{movie}), \\
    (\texttt{coworker},\texttt{coworker})
  \end{array} \right\}
}
\end{equation}
Node \texttt{director2} features two outgoing edges, one labeled \texttt{born\_in} to node \texttt{place}, the other labeled \texttt{directed} to \texttt{movie}.
Node \texttt{director} in $G_2$ dual simulates \texttt{director2}, since it has an outgoing edge with label \texttt{born\_in} to node \texttt{place}, and \texttt{place} in $G_2$ dual simulates \texttt{place} in $G_1$.
The same argument holds for node \texttt{movie}.
By following through the argumentation for every pair of nodes in \eqref{eq:dualsim-small-example}, it can be shown that $G_2$ indeed dual simulates $G_1$ under the indicated dual simulation \eqref{eq:dualsim-small-example}.
Observe that a single node, \eg \texttt{director}, may dual simulate more than one node.
\begin{table}
\caption{Summary of Symbols}
\label{tab:summary}
\centering
\scalebox{.98}{\begin{tabular}{ll}
\toprule
$G=(V,\Sigma,E)$ & edge-labeled directed graph \\
$\mathfrak{F}_G^a$, $\mathfrak{B}_G^a$ & forward/backward map for label $a$ in $G$ \\
$\db=(O_\db,\Sigma,E_\db)$ & graph database \\
\midrule
$\chi_S : V_1 \to 2^{V_2}$ & characteristic function for relation $S\subseteq V_1 \times V_2$ \\
\midrule
\Q, \Q[1], \Q[2] & \sparql queries \\
$\lbr\Q\rbr_\db$ & set of matches due to \sparql semantics \\
$\mu : \vars(\Q) \to O_\db$ & match to query \Q in \db \\
$\mu_1 \compat \mu_2$ & compatibility predicate between $\mu_1$ and $\mu_2$ \\
\midrule
$\mathcal E = (\mathtt{Var},\mathtt{Eq})$ & system of inequalities
\\\bottomrule
\end{tabular}}
\end{table}
\begin{definition}[Dual Simulation~\cite{Ma2014}]\label{def:dual-simulation}
  Let $G_i = (V_i, \Sigma, E_i)$ ($i=1,2$) be two graphs.
  A relation $S\subseteq V_1\times V_2$ is a {\em dual simulation between $G_1$ and $G_2$} iff for each $(v_1,v_2)\in S$,
  \begin{enumerate}[(i)]
    \item $(v_1, a, w_1) \in E_1$ implies $\exists w_2\in V_2 : (v_2, a, w_2)\in E_2$ and $(w_1,w_2)\in S$,
    \item $(u_1, a, v_1) \in E_1$ implies $\exists u_2\in V_2 : (u_2, a, v_2)\in E_2$ and $(u_1,u_2)\in S$.
  \end{enumerate}
  We say that $G_2$ {\em dual simulates} $G_1$ iff there is a non-empty dual simulation between $G_1$ and $G_2$.
\end{definition}
Note that the trivial dual simulation $S=\emptyset$ would certify that any two graphs are dual simulating each other.
In a graph query setting we call $G_1$ {\em pattern graph} and $G_2$ is the {\em graph database}.
Reconsider the introductory example query \X[1].
The graph in Fig.~\ref{fig:dualsim-patterns}(b) dual simulates the graph representation of \X[1] in Fig.~\ref{fig:database}(b).
A dual simulation is realized by ignoring node \texttt{place}.
Hence, not every node of the graph database has to participate in a dual simulation relation.
Furthermore, the graph in Fig.~\ref{fig:dualsim-patterns}(a) neither dual simulates nor is dual simulated by the graph in Fig.~\ref{fig:database}(b).
Regarding the graph database depicted in Fig.~\ref{fig:database}(a) and the graph representation of \X[1] in Fig.~\ref{fig:database}(b), dual simulation \eqref{eq:dualsim-example} turns out to be particularly useful in the upcoming sections.
%\todo{useful klingt im vorherigen satz irgendwie unpassend. ich würde in dem nächsten satz auch ganz kurz pruning erwähnen.}
%
\begin{equation}\label{eq:dualsim-example}
{\fontsize{7.7}{7.7}\selectfont
\left\{\begin{array}{l}
     (\texttt{director}, \texttt{B. De Palma}), (\texttt{director}, \texttt{G. Hamilton}),\\
     (\texttt{coworker}, \texttt{D. Koepp}), (\texttt{coworker}, \texttt{H. Saltzman}),\\
     (\texttt{movie}, \texttt{Mission: Impossible}), (\texttt{movie}, \texttt{Goldfinger})
\end{array}\right\}
}
\end{equation}
It comprises exactly the nodes of the two subgraphs from the result set of \X[1].
Instead of considering the full graph database (\ie Fig.~\ref{fig:database}(a)) we would ignore all graph database nodes but those mentioned by dual simulation \eqref{eq:dualsim-example}.
Computing this dual simulation is possible in {\scshape Ptime}~\cite{Ma2014}, as opposed to \sparql query evaluation being {\scshape Pspace}-complete~\cite{Perez2009,Schmidt2010FoundationsOptimization}.
How to perform this computation fast is subject to the next section.
We apply dual simulation principles to \sparql for query processing in Sect.~\ref{sec:pruning}.

\section{A Perspective on Dual Simulation}\label{sec:soi}\label{sec:query-semantics}
%
% -*- root: ../main.tex -*-
%
At the end of the last section we have seen a dual simulation between a graph representation of a \sparql query (BGP \X[1]) and a graph database (Fig.~\ref{fig:database}(a)), covering all nodes relevant for computing the result set of \X[1].
In Sect.~\ref{sec:pruning} we show that the existence of such a dual simulation is not coincidental, since every match for \sparql queries like \X[1] is contained in a maximal dual simulation (\cf Theorem~\ref{thm:bgp-soundness}).
A dual simulation $S$ is maximal iff there is no dual simulation $S'$ such that $S\subset S'$.
Fortunately, there is exactly one such maximal dual simulation between any two graphs, the {\em largest dual simulation}.
\begin{proposition}[Proposition 2.1~\cite{Ma2014}]\label{prop:maximal-dualsim}
  For any two graphs $G_1$ and $G_2$, there is a unique largest dual simulation $S_{\max}$ between $G_1$ and $G_2$, \ie for any dual simulation $S$ between $G_1$ and $G_2$, $S \subseteq S_{\max}$.
\end{proposition}
The proof exploits the fact that, whenever we have two dual simulations $S_1$ and $S_2$ between the graphs, their union $S_1\cup S_2$ is a dual simulation.
Incorporating dual simulation in graph pattern matching or \sparql query processing amounts to computing the largest dual simulation between an appropriate representation of the query and the graph database.
All graph database nodes captured by the largest dual simulation are relevant for answering the query.

Computing the largest (dual) simulation is the algorithmic basis for solving the graph (dual) simulation problem, \ie given two graphs $G_1$ and $G_2$, does $G_2$ (dual) simulate $G_1$.
To the best of our knowledge, all published algorithms for this task~\cite{Ma2014,Henzinger1995} work on the same principles.
Starting with the largest possible relation between the two node sets, the algorithms incrementally disqualify pairs of nodes violating Def.~\ref{def:dual-simulation}.
The procedures are guaranteed to terminate when no pair of nodes can be disqualified anymore.
Although the standard algorithms share an $\onot{|V_2|^3}$ data (runtime) complexity, we observed that these algorithms only allow for the naive evaluation strategy described above, which have originally been invented for comparing graphs of unknown sizes with each other.
The aforementioned data complexity follows from generalizing the existing algorithms \cite{Henzinger1995} and \cite{Ma2014} to edge-labeled graphs (\cf \rep{Sect.~\ref{sub:discussion}}\conf{our supplementary report~\cite{Mennicke2019TR}} for a detailed derivation).
This inflexibility generates high query running times that would easily be outperformed by state-of-the-art query evaluation, \eg by Virtuoso (\cf Sect.~\ref{sec:evaluation}).

Subsequently, we develop a novel solution which computes the largest dual simulation and exploits run-time analytics to dynamically adapt evaluation strategies.
Key to our solution is the reformulation of the algorithm as a system of inequalities which allows for two dynamically interchangeable evaluation strategies.
Although the worst-case complexity of our solution remains unaltered (\cf Sect.~\ref{sub:discussion}), compared to the existing algorithms, we gain a degree of freedom allowing for a systematic reduction of iterations to eventually reach the largest dual simulation (\cf Sect.~\ref{sub:discussion}).
% Since queries are usually assumed to be much smaller than the database, the organizational overhead in memory is negligible.
As we show in Sect.~\ref{sec:evaluation} the new procedure shows extremely low computation times, a solid basis for query processing.
Our solution is engineered in three steps.
First, we define a set of inequalities equivalent to the coinductive definition of dual simulation in Def.~\ref{def:dual-simulation}.
We further show how to derive a fast implementation based on bit-vectors and bit-matrices.
Last, we provide a discussion on optimizations realized in our software prototype\footnote{available at GitHub \url{https://github.com/ifis-tu-bs/sparqlSim}}.

\subsection{Groundwork}\label{sub:groundwork}
Any binary relation $R\subseteq A\times B$, over sets $A$ and $B$, has a characteristic function $\chi_R : A \to 2^{B}$ with $\chi_R(a) := \{ b\in B \mid (a,b)\in R \}$.
For a dual simulation $S$ between graphs $G_1$ and $G_2$, $\chi_S$ associates with each node $v\in V_1$ a set of dual simulating nodes $\chi_S(v)\subseteq V_2$.
Consider an edge $(v, a, w)$ of $G_1$ and node $v'\in \chi_S(v)$.
If $S$ is a dual simulation, then for $\chi_S$ we derive
\begin{equation}\label{eq:implication}
   \exists w' : (v',a,w')\in E_2 \text{~and~} w'\in \chi_S(w)\text.
\end{equation}
The problem with \eqref{eq:implication} is that there may be many $w'$ qualifying for $(v',a,w')\in E_2$ but $w'\notin\chi_S(w)$.
We pursue to have a single operation allowing us to quickly verify the existence of $w'$.
Therefore, recall that for any graph, here graph database $G_2$, we have a forward adjacency map $\mathfrak{F}_{G_2}^a$ for each label $a\in\Sigma$ (\cf Sect.~\ref{sec:preliminaries}).
By exploiting these maps we prove existence of a $w'$ in \eqref{eq:implication} simply by intersecting the row of $v'$ in $\mathfrak{F}_{G_2}^a$ and the nodes simulating $w$, \ie
\begin{equation}\label{eq:set-implication}
  \mathfrak{F}_{G_2}^a(v')\cap \chi_S(w) \neq \emptyset\text.
\end{equation}
\eqref{eq:set-implication} still only checks for one pair of nodes $(v,v')$.
Combining this equation for all $v'\in \chi_S(v)$ yields
\begin{equation}\label{eq:set-combined1}
  \begin{array}{rlcl}
    \bigwedge_{v'\in \chi_S(v)} & \mathfrak{F}_{G_2}^a(v') \cap \chi_S(w) \neq \emptyset\text.
  \end{array}
\end{equation}
The same encoding applies to Def.~\ref{def:dual-simulation}(ii), this time using the backward map,
\begin{equation}\label{eq:set-combined2}
  \begin{array}{rlcl}
    \bigwedge_{w'\in \chi_S(w)} & \mathfrak{B}_{G_2}^a(w') \cap \chi_S(v) \neq \emptyset\text.
  \end{array}
\end{equation}
The combination of both equations \eqref{eq:set-combined1} and \eqref{eq:set-combined2} yields two inequalities equivalent to the definition of dual simulation and the key for our efficient implementation.
\begin{lemma}\label{lemma:inequality}
  Let $G_1=(V_1,\Sigma,E_1)$ and $G_2=(V_2,\Sigma,E_2)$ be graphs with $(v,a,w)\in E_1$.
  For a binary relation $S\subseteq V_1\times V_2$ satisfying \eqref{eq:set-combined1} and \eqref{eq:set-combined2}, it holds that \eqref{eq:subset-combined} is satisfied.
  \begin{equation}\label{eq:subset-combined}
    \begin{array}{rcrclcl}
      (i) && \chi_S(w) & \subseteq & \bigcup_{v'\in \chi_S(v)} \mathfrak{F}_{G_2}^a(v') && \text{and} \\
      (ii) && \chi_S(v) & \subseteq & \bigcup_{w'\in \chi_S(w)} \mathfrak{B}_{G_2}^a(w')
    \end{array}
  \end{equation}
\end{lemma}
\begin{proof}
  \Wlog, we show inequality (i) only.
  Inequality (ii) is completely analogous.
  Towards a contradiction assume $\chi_S(w)\not\subseteq \bigcup_{v'\in\chi_S(v)} \mathfrak{F}_{G_2}^a(v')$.
  Hence, there is a $w'\in\chi_S(w)$ such that for each $v'\in\chi_S(v)$, $w'\notin \mathfrak{F}_{G_2}^a(v')$, \ie $(v',a,w')\notin E_2$.
  As a consequence, $\chi_S(v)$ and $\mathfrak{B}_{G_2}^a(w')$ are disjoint for each $v'\in \chi_S(v)$, contradicting our assumption that \eqref{eq:set-combined2} holds.
  Therefore, such a $w'$ cannot exist, allowing to conclude that $\chi_S(w)\subseteq \bigcup_{v'\in\chi_S(v)} \mathfrak{F}_{G_2}^a(v')$.
\end{proof}
Phrased differently, dual simulations $S$ satisfy \eqref{eq:subset-combined} for every edge $(v, a, w)$ of $G_1$.
Lemma~\ref{lemma:inequality} reveals an important observation that, to the best of our knowledge, has not been published so far:
The reason why \eqref{eq:subset-combined} holds is that part (ii) prevents part (i) from getting ill-formed and vice versa.
The fast algorithm we obtain here is a consequence of the duality in dual simulation.
Conversely, every solution to \eqref{eq:subset-combined} is a dual simulation.
\begin{proposition}\label{prop:soi-correctness}
  Let $G_1$ and $G_2$ be graphs.
  $S\subseteq V_1 \times V_2$ is a dual simulation between $G_1$ and $G_2$ iff for every edge $(v, a, w)\in E_1$, \eqref{eq:subset-combined} holds for $S$.
\end{proposition}
\begin{report}
\begin{proof}
  The implication, \ie a dual simulation $S$ satisfies \eqref{eq:subset-combined}, is analogous to the proof of Lemma~\ref{lemma:inequality}.
  Therefore, assume that one of the inequalities is not satisfied and conclude the assumption that $S$ is a dual simulation is violated.

  Conversely, assume we have $S \subseteq V_1 \times V_2$ such that \eqref{eq:subset-combined} holds for every $(v,a,w)\in E_1$.
  We prove $S$ to be a dual simulation.
  Let $(v,v')\in S$, \ie $v'\in\chi_S(v)$, and $(v,a,w)\in E_1$.
  We need to show that there is a $w'$ such that $(v',a,w')\in E_2$ and $(w,w')\in S$.
  From \eqref{eq:subset-combined}(ii) we get that for some $w'\in\chi_S(w)$ we have that $v'\in\mathfrak{B}_{G_2}^a(w')$.
  This $w'$ completes the proof, since (1) from $v'\in\mathfrak{B}_{G_2}^a(w')$ follows $(v',a,w')\in E_2$ and (2) from $w'\in\chi_S(w)$, we get that $(w,w')\in S$.
  Case $(u,a,v)\in E_1$ is completely analogous.
\end{proof}
\end{report}
\begin{conference}
The proof builds on the principles of Lemma~\ref{lemma:inequality} and can be found in our report~\cite{Mennicke2019TR}.
\end{conference}
Hence, \eqref{eq:subset-combined} characterizes dual simulations, and we can use it to compute the largest dual simulation.
The algorithm works as follows.
We begin with $S_0 := V_1 \times V_2$.
For each edge of $G_1$, check whether \eqref{eq:subset-combined} is satisfied by $S_0$.
Assume \eqref{eq:subset-combined}(i) fails for an edge $(v, a, w)$.
Then, $S_{1}$ is computed by $\chi_{S_{1}}(u):=\chi_{S_0}(u)$ for $u\neq w$ and $\chi_{S_{1}}(w) := \chi_{S_0}(w) \cap \bigcup_{v'\in \chi_{S_0}(v)} \mathfrak{F}_{G_2}^a(v')$.
We get rid of all non-simulating nodes of $w$ relative to $S_0$ in a single iteration.
This procedure is repeated for $S_1, S_2, \ldots$ until we reach an $S_k$ satisfying \eqref{eq:subset-combined} for every edge of $G_1$.

Even though we maintain the {\scshape Ptime} nature of other algorithms (\cf Sect.~\ref{sub:discussion}), we still miss a way to quickly compute $\bigcup_{v'\in \chi_{S}(v)} \mathfrak{F}_{G_2}^a(v')$ and access $\chi_S(v)$.
Therefore, the forthcoming implementation works with bit-representations of $\chi_S(v)$ and $\mathfrak{F}_{G_2}^a, \mathfrak{B}_{G_2}^a$, paving the way for optimization in time- and space-consumption (\eg \cite{Atre2010}).
In that setting we derive a {\em system of inequalities} (SOI) from Prop.~\ref{prop:soi-correctness}, for which dual simulations $S$ serve as valid assignments.

\subsection{Engineering}\label{sub:implementation}
Our goal is to obtain the facilities for achieving a fast implementation of dual simulation processing.
Recall that we need to compute the largest dual simulation and we do this by a {\em system of inequalities} according to \eqref{eq:subset-combined}.
The challenge is to find a way to quickly compute the unions
\begin{equation}\label{eq:unions}
  \begin{array}{ccc}
    \bigcup_{v'\in \chi_S(v)} \mathfrak{F}_{G_2}^a(v') & \text{and} & \bigcup_{w'\in \chi_S(w)} \mathfrak{B}_{G_2}^a(w')\text.
  \end{array}
\end{equation}
Combinations of vectors and matrices, especially when encoding information only bit-wise, promise fast computations.
Hence, we interpret the adjacency maps of $G_2$ as adjacency bit matrices.
Reconsider the graph in Fig.~\ref{fig:dualsim-patterns}(a).
For label \texttt{born\_in}, this graph provides two adjacency matrices,
\begin{center}
  \scalebox{.72}{
$\begin{array}{ccc}
  \mathfrak{F}_{Fig.~\ref{fig:dualsim-patterns}(a)}^{\texttt{born\_in}} = \left( \begin{array}{ccccc}
    0 & 0 & 0 & 0 & 0 \\
    1 & 0 & 0 & 0 & 0 \\
    1 & 0 & 0 & 0 & 0 \\
    0 & 0 & 0 & 0 & 0 \\
    0 & 0 & 0 & 0 & 0
  \end{array} \right) & \text{and} & \mathfrak{B}_{Fig.~\ref{fig:dualsim-patterns}(a)}^{\texttt{born\_in}} = \left( \begin{array}{ccccc}
    0 & 1 & 1 & 0 & 0 \\
    0 & 0 & 0 & 0 & 0 \\
    0 & 0 & 0 & 0 & 0 \\
    0 & 0 & 0 & 0 & 0 \\
    0 & 0 & 0 & 0 & 0
  \end{array} \right)\text.
\end{array}$}
\end{center}
Here, we assume the set of nodes of graphs to be ordered by some pre-defined index, \eg
%
%\begin{center}
$v_1 = \texttt{place}$, $v_2 = \texttt{director1}$, $v_3 = \texttt{director2}$, $v_4 = \texttt{coworker}$, and $v_5 = \texttt{movie}$.
%\end{center}
%
Also, $\chi_S$ can be seen as a matrix with $k=|V_{1}|$ rows, one for each node of pattern graph $G_1$, and $n=|V_{2}|$ columns.
Specifically, for a dual simulation $S$ a $1$ in position $(i,j)$ means that the $i^{\text{th}}$ node of the pattern graph is simulated by the $j^{\text{th}}$ node of the graph database.
For ease of presentation, the pattern graph $G_1$ does not have an indexed node set.
Consequently, for node $v$ of the pattern graph and $j\leq n$, we access the $j^{\text{th}}$ component of $v$'s row by $\chi_S(v,j)$.
By $\chi_S(v)$ we get $v$'s row vector sliced from matrix $\chi_S$.
The desired unions \eqref{eq:unions} are now achieved by bit-matrix multiplications\footnote{For vector $\mathfrak{v}$ and matrix $\mathfrak{A}$, $\mathfrak{v} \times_b \mathfrak{A} = \mathfrak{w}$ where $\mathfrak{w}(j)=1$ iff there is an $i$ such that $\mathfrak{v}(i)=1$ and $\mathfrak{A}(i,j)=1$.} (symbol $\times_b$),
\begin{equation}\label{eq:bit-unions}
  \begin{array}{ccc}
    \chi_S(v) \times_b \mathfrak{F}_{G_2}^a & \text{and} & \chi_S(w) \times_b \mathfrak{B}_{G_2}^a\text.
  \end{array}
\end{equation}
The result of the multiplication is the {\em reachable nodes via $a$-labeled (forward) edges} from any simulating node of $v$.
For instance, assume that $\chi_S(\texttt{director})=\chi_S(\texttt{place})=(1,1,1,1,1)$.
Then, for edge $(\texttt{director},\texttt{born\_in},\texttt{place})$:\\[.8em]
% \begin{center}
  \centerline{$\begin{array}{lclcl}
    \chi_S(\texttt{director}) & \times_b & \mathfrak{F}_{Fig.~\ref{fig:dualsim-patterns}(a)}^{\texttt{born\_in}} & = & (1,0,0,0,0) = r_1 \\
    \chi_S(\texttt{place}) & \times_b & \mathfrak{B}_{Fig.~\ref{fig:dualsim-patterns}(a)}^{\texttt{born\_in}} & = & (0,1,1,0,0) = r_2\text.
  \end{array}$}\\[.8em]
% \end{center}
%
Hence, $r_1$ reveals that only node $\texttt{place}$ is reachable via forward edges labeled \texttt{born\_in}.
Conversely, by $\texttt{born\_in}$-labeled backward edges we reach \texttt{director1} as well as \texttt{director2}.
The results are used to update a given relation $S$, according to \eqref{eq:subset-combined}.
In the example above, $r_2$ shows that $\chi_S(\texttt{director})\neq (1,1,1,1,1)$, since the only reachable nodes are \texttt{director1} and \texttt{director2}.
Thus, $\chi_S(\texttt{director})=(1,1,1,1,1)\not\leq (0,1,1,0,0) = r_2$, but according to Prop.~\ref{prop:soi-correctness}, a dual simulation $S$ satisfies \eqref{eq:subset-combined}, now possible to formulate by bit-matrix operations for edges $(v,a,w)\in E_{G_1}$,
\begin{equation}\label{eq:bit-soi}
  \begin{array}{rclclcl}
    \chi_S(w) & \leq & \chi_S(v) & \times_b & \mathfrak{F}_{G_2}^a && \text{and} \\
    \chi_S(v) & \leq & \chi_S(w) & \times_b & \mathfrak{B}_{G_2}^a\text.
  \end{array}
\end{equation}
After observing the wrong value of $\chi_S(\texttt{director})$ we update relation $S$ to $S'$ by $\chi_{S'}(\texttt{director}) := \chi_S(\texttt{director}) \wedge r_2$ (component-wise conjunction of the two vectors).
This enables us to give an algorithm for the dual simulation problem between two graphs $G_1$ and $G_2$ as a solution of the system of inequalities $\mathcal E = ( \mathsf{Var}, \mathsf{Eq} )$, where every node $v$ of the graph pattern is a variable, \ie $\mathsf{Var}:=V_{1}$, and $\mathsf{Eq}$ contains for each pattern edge $(v,a,w)\in E_{1}$, the following equations:
\begin{equation}\label{eq:soi}
  \begin{array}{rclcrcl}
    w & \leq & v \times_b \mathfrak{F}_{G_2}^a & \text{and} &
    v & \leq & w \times_b \mathfrak{B}_{G_2}^a\text.
  \end{array}
\end{equation}
Fig.~\ref{fig:example-soi} shows the SOI for computing dual simulations for the graphs in Fig.~\ref{fig:dualsim-patterns}(a) and (b).
Assignments to the variables $v,w\in\mathsf{Var}$ are relations $S\subseteq V_{1}\times V_{2}$.
The algorithm computing the largest dual simulation between $G_1$ and $G_2$ proceeds as follows.
\begin{enumerate}
  \item Set $S_0 := V_{1} \times V_{2}$ and all inequalities in $\mathsf{Eq}$ {\em unstable}.%\footnote{The algorithm proceeds until all equations are stable.}.
  \item Let $S_i$ be the current candidate relation.
        Pick any unstable inequality $\epsilon\in\mathsf{Eq}$
  \begin{enumerate}
    \item If $S_i$ is valid for $\epsilon$, set $\epsilon$ stable and continue with (2).
    \item If $S_i$ is invalid for inequality $\epsilon = v \leq w \times_b \mathfrak{A}$ (for $\mathfrak{A}\in \{ \mathfrak{F}_{G_2}^a, \mathfrak{B}_{G_2}^a \mid a \in \Sigma \}$), then $\chi_{S_i}(w)\times_b \mathfrak{A}=r$ and $\chi_S(v)\not\leq r$.
    Update $S_i$ to $S_{i+1}$ such that\\[.3em]
    \centerline{
    $\chi_{S_{i+1}}(x) := \left\{
    \begin{array}{lcl}
      \chi_{S_i}(x) \wedge r && \text{if } x = v \text{ and} \\
      \chi_{S_i}(x) && \text{otherwise.}
    \end{array} \right.$}\\[.3em]
    Furthermore, every inequality $y \leq v \times \mathfrak{A} \in\mathtt{Eq}$ are reset to unstable.
    Mark $\epsilon$ stable and continue with (2).
  \end{enumerate}
\end{enumerate}
The initialization step of $S_0$ can also be expressed in terms of inequalities, in that for every pattern node $v$, we add \eqref{eq:init1} to the set of inequalities $\mathsf{Eq}$.
\begin{equation}\label{eq:init1}
  \begin{array}{rcl}
    v & \leq & \underline{\bf{1}}
  \end{array}
\end{equation}
$\underline{\bf{1}}$ is the vector containing a $1$ in every component.
The dual simulation given by \eqref{eq:dualsim-small-example} is the largest solution to the SOI in Fig.~\ref{fig:example-soi}, thus it constitutes the largest dual simulation.

\begin{figure}[tbp]
    \begin{center}
      % \begin{tabular}{lcl}
        \scalebox{.8}{$\begin{array}{rcl}
              \texttt{place} & \leq & \texttt{director1} \times_b \mathfrak{F}_{Fig.~\ref{fig:dualsim-patterns}(b)}^{\texttt{born\_in}} \\
              \texttt{place} & \leq & \texttt{director2} \times_b \mathfrak{F}_{Fig.~\ref{fig:dualsim-patterns}(b)}^{\texttt{born\_in}} \\
              \texttt{director1} & \leq & \texttt{place} \times_b \mathfrak{B}_{Fig.~\ref{fig:dualsim-patterns}(b)}^{\texttt{born\_in}} \\
              \texttt{director2} & \leq & \texttt{place} \times_b \mathfrak{B}_{Fig.~\ref{fig:dualsim-patterns}(b)}^{\texttt{born\_in}} \\
            %\end{array}$} &&
            %\scalebox{.7}{$\begin{array}{rcl}
              \texttt{coworker} & \leq & \texttt{director1} \times_b \mathfrak{F}_{Fig.~\ref{fig:dualsim-patterns}(b)}^{\texttt{worked\_with}} \\
              \texttt{director1} & \leq & \texttt{coworker} \times_b \mathfrak{B}_{Fig.~\ref{fig:dualsim-patterns}(b)}^{\texttt{worked\_with}} \\
              \texttt{movie} & \leq & \texttt{director2} \times_b \mathfrak{F}_{Fig.~\ref{fig:dualsim-patterns}(b)}^{\texttt{directed}} \\
              \texttt{director2} & \leq & \texttt{movie} \times_b \mathfrak{B}_{Fig.~\ref{fig:dualsim-patterns}(b)}^{\texttt{directed}}
            \end{array}$}
      % \end{tabular}
    \end{center}
  \caption{System of Inequalities Characterizing Largest Dual Simulation between Fig.~\ref{fig:dualsim-patterns}(a) and (b)}\label{fig:example-soi}
\end{figure}

\subsection{Complexity and Optimization}\label{sub:discussion}
Initializing $S_0$ takes time \onot{|V_1| \cdot |V_2|} in a naive implementation.
We execute step 2) at most $|V_1|\cdot |V_2|$ times, since there are $|V_1|$ pattern nodes for which at most $|V_2|$ data nodes can be disqualified.
Let $\epsilon = v\leq w \times_b \mathfrak{A}$ be in $\mathsf{Eq}$ and $S_i$ the current candidate relation.
Computing $r=\chi_{S_i}(w)\times_b \mathfrak{A}$ is in $\onot{|V_2|^2}$ time.
By further regarding the intersection $\chi_{S_i}(v)\wedge r$, we obtain an overall time complexity of $\onot{|V_2|^2 + |V_2|}=\onot{|V_2|^2}$ for updating $S_i$ to $S_{i+1}$ which validates $\epsilon$.
For every edge in $G_1$ we have two equations in $\mathsf{Eq}$, \ie $|\mathsf{Eq}| = \onot{|E_1|}$.
Thus, assuming pattern $G_1$ and data graph $G_2$ as input, our algorithm has a combined complexity of $\onot{(|V_1|\cdot |V_2|)\cdot|E_1|\cdot|V_2|^2}$.
In terms of data complexity we have a worst-case runtime of $\onot{|V_2|^3}$, virtually the same complexity as of any other dual simulation algorithm\conf{ (\cf~\cite{Mennicke2019TR})}.

\begin{report}
The acquired combined complexity of our solution is higher than that of an algorithm leveraging by Henzinger \etal's so-called {\em HHK algorithm}, being in $\onot{mn}$ time.
Note that HHK assumes a single node-labeled graph (\ie no labels on the edges) $G=(V,E)$ and $|V|=n < m=|E|$ ($m\leq n^2$).
Dual simulation requires the execution of HHK two times, which leaves the overall complexity invariant.
However, considering a separation into pattern and data graph as well as adding edge labels does have an effect on the resulting HHK adaptation.
Separating pattern $G_1=(V_1,E_1)$ from data graph $G_2=(V_2,E_2)$ yields an overall combined complexity of $O(|E_2|\cdot |V_2|)$, \ie the runtime complexity solely depends on the data graph $G_2$.
The crux of HHK is an additional data structure which tracks for each pattern node $v$, the set of adjacent data nodes from which the simulating nodes of $v$ are unreachable.
These are the definite nodes that cannot simulate the respective adjacent nodes.
The maintenance of these {\em removal sets} is the key component in the complexity analysis~\cite{Henzinger1995}.
Remarkably, combined and data complexity are equal for HHK for unlabeled graphs.
If, additionally, edge labels are considered, the runtime estimation alters at least to $O(|\Sigma(G_1)|\cdot|V_2|^3)$, where $\Sigma(G_1)$ denotes the set of actually used labels in $G_1$.
This is because every update of a removal set requires only a single adjacency matrix of size $O(|V_2|^2)$.
However, for every label on the incident edge of a node, there is a different removal set to maintain.

Due to the graph query setting, there is no difference in worst-case data complexity between HHK and our solution.
In fact, the algorithm of Ma et al.~\cite{Ma2014}, adjusted to labeled graphs, enjoys the same data complexity, \ie $\onot{|V_2|^3}$.
As a consequence, we formulate the {\em specific data complexity hypothesis for dual simulation graph query processing}:
The real computation times of naive implementations of HHK and the algorithm of Ma \etal should show no significant differences in the (labeled) graph query setting.
We provide experimental evidence for this hypothesis in Sect.~\ref{sec:evaluation}.
Although the existence of suitable algorithmic tweaks for any of the abovementioned algorithms is not deniable, we advertise our algorithmic framework for its separation into algorithmic representation as a system of inequalities and evaluation algorithm, externally adaptable by static and dynamic heuristics.
\end{report}

\begin{report}
An immediate optimization is given by altering the initial relation $S_0$, syntactically exploiting that for a variable/node $v$ in $G_1$, candidate nodes are only those supporting incident edges of $v$.
Therefore, let us denote by $\mathfrak{f}_{G_2}^a$ the bit-vector that summarizes the rows of $\mathfrak{F}_{G_2}^a$ in that $\mathfrak{f}_{G_2}^a(i)=1$ if there is a $j$ with $\mathfrak{F}_{G_2}^a(i,j)=1$, and $\mathfrak{f}_{G_2}^a(i)=0$ otherwise.
In the same lines, $\mathfrak{b}_{G_2}^a$ is defined as the summary of $\mathfrak{B}_{G_2}^a$.
Then for each variable/node $v$ in $G_1$, we replace inequality~(\ref{eq:init1}) by
\begin{equation}\label{eq:init}
  \begin{array}{rcl}
    v & \leq & \bigwedge_{(v,a,w)\in E_1} \mathfrak{f}_{G_2}^a \wedge \bigwedge_{(u,a,v)\in E_1} \mathfrak{b}_{G_1}^a\text.
  \end{array}
\end{equation}
\end{report}

Our characterization of dual simulation and its implementation open up dynamic evaluation strategies for the constructed SOI.
First, the order in which the equations are evaluated has an impact on the overall runtime.
For our experiments, we have chosen an order that aims at shrinking the simulation as early as possible, \eg by preferring inequalities with matrix components having more empty columns, which indicates sparsity of the respective matrices.
Second, the computation of $r$ (step 2b of the algorithm) may be performed row-wise or column-wise.
Again, we follow the strategy of fewer iterations, \ie in $v\leq w \times_b \mathfrak{A}$ we choose a row-wise evaluation if and only if $\chi_S(w)$ has fewer bits set than $\chi_S(v)$.
As it turns out (\cf Sect.~\ref{sub:eval-discussion}) there is not a single heuristic that fits all input patterns and databases.

Our proof-of-concept implementation keeps $G_2$ in memory by its adjacency matrices.
$G_1$ is stored by its system of inequalities, including $\onot{|V_1|}$ bit-vectors representing $\chi_S$.
For every graph pattern $G_1$, it suffices to load those adjacency matrices that are needed the pattern.
Hence, the worst-case memory consumption is determined by the graph pattern and by the adjacency matrix requiring the most memory.
Note that due to bit-vector storage techniques, such as gap-length encoding, the worst memory consumption might not occur with the label storing the most bits.
Combined with the memory-economical implementation by Atre \etal~\cite{Atre2010,Atre2015LeftDescriptors} we are quite optimistic that our implementation may directly be used within the preprocessing step of the {\em BitMat} tool set.
Our dual simulation processing applied to \sparql queries yields decent pruning factors (\cf Sect.~\ref{sec:evaluation}), significantly improving upon those reported by Atre~\cite{Atre2015LeftDescriptors}.

\section{Dual Simulation for \sparql}\label{sec:subgraph-selection}\label{sec:pruning}
%
% -*- root: ../main.tex -*-
%
Having clarified the foundational and algorithmic aspects of dual simulations we now approach an actual query language, namely \sparql.
We choose \sparql for its high-quality standardization by the W3C~\cite{Prudhommeaux2008} and its extensive formal treatment, \eg \cite{Perez2009,Schmidt2010FoundationsOptimization,Arenas2013,Arenas2017}.
\begin{report}
Although \sparql~1.1 has been around for some time, the fundamental properties of the query language remain the same as for \sparql~1.0.
We are aware of the recent report on the semantic foundation of the \mbox{Neo4J} query language \cypher~\cite{neo4jsemantics2018}, and confident about the wider applicability of the forthcoming techniques to this language.
\end{report}
Subsequently, for \sparql's least complex construct we canonically obtain dual simulation processing respecting all matches any \sparql query processor would find.
We further discuss \sparql's join operators.
For each query language feature we obtain a soundness result guaranteeing that the original \sparql matches are preserved for further processing.

\subsection{Basic Graph Patterns}\label{sub:bgp}
As for RDF, {\em triple patterns} are first-class citizens of \sparql.
For the presentation of the upcoming material, we assume subject and object of a triple $t=(s,p,o)$ to be variables from an infinite domain of variables $\mathcal V$, ranging over by $\mathtt{v}, \mathtt{v_1}, \mathtt{v_2}, \ldots$.
A variable $\mathtt{v_1}$ is usually introduced by a leading question mark, \ie {\tt ?$\mathtt{v_1}$} (\cf \X[1]).
In formal notation, however, we drop this syntactic convention and write $\mathtt{v_1}$.

Querying a graph database $\db=(O_\db,\Sigma,E_\db)$ yields a set of partial mappings from the set of variables to actual database objects.
For instance, the single triple pattern $t=(\mathtt{v_1},\texttt{population},\mathtt{v_2})$ gives rise to a match identifying $\mathtt{v_1}$ with node \texttt{Saint Join} and $\mathtt{v_2}$ with $\texttt{70.063}$ (\cf Fig.~\ref{fig:database}(a)).
By $\vars(t)$ we denote the set of variables occurring in triple $t$, \ie $\vars(t)=\{ \mathtt{v_1}, \mathtt{v_2} \}$ for the abovementioned $t$.
A {\em candidate in \db} is a partial function $\mu : \mathcal V \to O_\db$.
$\dom(\mu)$ denotes the set of variables for which candidate $\mu$ is defined.
A candidate $\mu$ is a {\em match for triple $t$ in \db} iff $\dom(\mu)=\vars(t)$ and, assuming $t=(\mathtt{v_1},a,\mathtt{v_2})$, $(\mu(\mathtt{v_1}), a, \mu(\mathtt{v_2}))\in E_\db$, abbreviated by $\mu(t)\in\db$.

We call sets of triple patterns $\mathbb G$ {\em basic graph patterns} (BGPs).
Function \vars and thereupon the notion of matches extend to BGPs by $\vars(\mathbb G)=\bigcup_{t\in \mathbb G} \vars(t)$, and $\mu$ is a match for $\mathbb G$ iff $\mu$ is a match for all triples $t\in\mathbb G$.
The result set $\lbr\mathbb G\rbr_\db$ for $\mathbb G$ \wrt \db contains all matches for $\mathbb G$ in \db.
Every BGP $\mathbb G$ can be seen as a graph $G(\mathbb G)=(V_{\mathbb G},\Sigma,\mathbb{G})$ by taking the set of variables occurring in $\mathbb G$ as set of nodes, \ie $V_{\mathbb G} := \{ \mathtt{v}, \mathtt{w} \mid (\mathtt{v},a,\mathtt{w})\in \mathbb G \}$.
The graph in Fig.~\ref{fig:database}(b) represents such a conversion of of query \X[1].

For dual simulation processing of a BGP $\mathbb{G}$ \wrt \db, we compute the largest dual simulation between $G(\mathbb G)$ and \db.
This procedure is sound in that every match $\mu$ for $\mathbb{G}$ in \db is a dual simulation and therefore must be contained in the largest dual simulation.
\begin{lemma}\label{lemma:match-is-sim}
  Let \db be a graph database and $\mathbb G$ be a BGP.
  Each $\mu\in\lbr\mathbb G\rbr_\db$ is a dual simulation between $G(\mathbb G)$ and \db.
\end{lemma}
\begin{report}
\begin{proof}
  We show that $\mu$ is a dual simulation between $G(\mathbb G)$ and \db.
  Let $(\mathtt{v},o)\in \mu$, \ie $\mu(\mathtt{v})=o$, and let $t\in\mathbb G$ such that $\mathtt{v}\in\vars(t)$.
  There are two cases to distinguish, for some $a\in\Sigma$ and $\mathtt{w},\mathtt{u}\in\vars(\mathbb G)$, \begin{inparaenum}[(a)]
    \item $t=(\mathtt{v}, a, \mathtt{w})$ and
    \item $t=(\mathtt{u}, a, \mathtt{v})$.
  \end{inparaenum}
  Since case (b) is completely analogous, we consider only (a).
  As $\mu$ is a match for $\mathbb{G}$, it is a match for $t$, \ie there is exactly one $o'=\mu(\mathtt{w})$ and $(o,a,o')\in E_\db$.
  %Hence, $\mu'$ is a function, which means that there is exactly one $o'\in O_\db$ with $(\mathtt{w},o')\in\mu'$ for all $t\in\mathbb{G}$, $\mu(t)\in\db$ implies that $\mu'(t)\in\db$ by construction.
  % Hence, $(o,a,o')\in E_{\db}$.
  %In summary, for $\mu'(\mathtt{v})= o$ and $(\mathtt{v},a,\mathtt{w})\in \mathbb{G}$, there is an $o'$ such that $(o,a,o')\in\mathbb{G}$ and $\mu'(\mathtt{w})=o'$, meeting the requirement Def.~\ref{def:dual-simulation}(i) for $\mu'$ to be a dual simulation.
  Hence, $o'$ meets the requirements of Def.~\ref{def:dual-simulation}(i).
\end{proof}
\end{report}
\begin{conference}
The reason why Lemma~\ref{lemma:match-is-sim} holds is that every match essentially constitutes a graph homomorphism.
We provide the detailed proof in our technical report~\cite{Mennicke2019TR}.
\end{conference}
The nodes disqualified by the largest dual simulation are irrelevant for any further query processing, obeying the original \sparql semantics.
\begin{theorem}\label{thm:bgp-soundness}
  Let \db be a graph database, $\mathbb G$ a BGP and $S$ the largest dual simulation between $G(\mathbb G)$ and \db.
  For each database node $o\in O_\db$ such that there are $\mathtt{v}\in\vars(\mathbb G)$ and $\mu\in\lbr\mathbb G\rbr_\db$ with $\mu(\mathtt{v})=o$, it holds that $(\mathtt{v},o)\in S$.
\end{theorem}
\begin{report}
\begin{proof}
%  Since $\mu$ is a match for $\mathbb{G}$, by Lemma~\ref{lemma:match-is-sim}, it holds that $\mu$ is a dual simulation.
  %Recall that $\mu$ is a dual simulation (Lemma~\ref{lemma:match-is-sim}).
  Towards a contradiction, assume there is a database node $o$ relevant to match variable $\mathtt{v}$ by $\mu\in\lbr\mathbb{G}\rbr_\db$ with $(\mathtt{v},o)\notin S$.
  But then $S\cup \mu$ is a dual simulation larger than $S$, contradicting the assumption that $S$ is the largest one.
\end{proof}
\end{report}
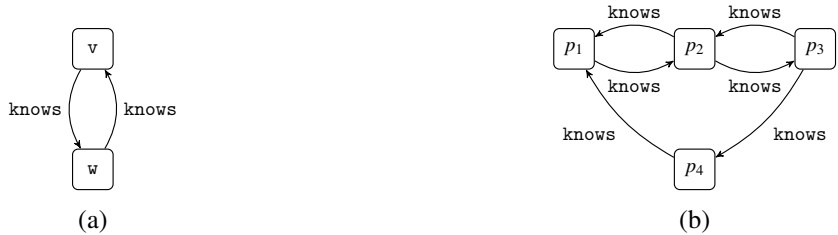
\begin{figure}[tbp]
\centering
  \begin{subfigure}[b]{.3\linewidth}
    \centering
    \scalebox{.7}{\begin{tikzpicture}[node distance=1.5]
      \node[entity] (p1) {\texttt{v}};
      \node[entity,below=of p1] (p2) {\texttt{w}}
        edge[pre,bend left=30] node[auto]{\texttt{knows}} (p1)
        edge[post,bend right=30] node[auto,swap]{\texttt{knows}} (p1);
    \end{tikzpicture}}
    \caption{}
  \end{subfigure}
  \quad
  \begin{subfigure}[b]{.64\linewidth}
    \centering
    \scalebox{.7}{\begin{tikzpicture}[node distance=1.5]
      \node[entity] (p2) {$p_2$};
      \node[entity,left=of p2] (p1) {$p_1$}
        edge[pre,bend left=30] node[auto] {\texttt{knows}} (p2)
        edge[post,bend right=30] node[auto,swap] {\texttt{knows}} (p2);
      \node[entity,right=of p2] (p3) {$p_3$}
        edge[pre,bend left=30] node[auto] {\texttt{knows}} (p2)
        edge[post,bend right=30] node[auto,swap] {\texttt{knows}} (p2);
      \node[entity,below=of p2] (p4) {$p_4$}
        edge[post,bend left=15] node[auto]{\texttt{knows}} (p1)
        edge[pre,bend right=15] node[auto,swap]{\texttt{knows}} (p3);
    \end{tikzpicture}}
    \caption{}
  \end{subfigure}
  \caption{(a) Graph Pattern $P$ and (b) Graph Database $K$, an example adapted from Ma \etal~\cite{Ma2014}}\label{fig:transitivity}
\end{figure}
Unfortunately, the converse, \ie irrelevant nodes for BGP result sets are ruled out by the largest dual simulation, does not hold in general.
Consider the example graphs $P$ and $K$ depicted in Fig.~\ref{fig:transitivity}(a) and (b).
The largest dual simulation between $P$ and $K$ includes node $p_4$ which is, however, not belonging to any match for the respective BGP.
%\todo[Rev 3]{What does "obligations" mean?}
The reason why $p_4$ must not be disqualified for variable/node $\mathtt{v}$ is that nodes $p_1$ and $p_3$ distribute the obligations for simulating variable/node $\mathtt{w}$.
Informally, $p_1$ knows $p_4$ via $p_2$ and $p_3$, although $p_1$ and $p_4$ do not have a direct link to one another.
\begin{report}
Non-transitive relationships sometimes appear transitive under dual simulation.
As long as acyclic queries are concerned, our process is also sound.
% Since this class of queries is rather small, we are not formally justifying this statement.
\end{report}

We compute the largest dual simulation by the largest solution of the SOI constructed from $G(\mathbb G)$ (\cf Sect.~\ref{sec:soi}).
From Theorem~\ref{thm:bgp-soundness} we learn the desirable property for systems of inequalities $\mathcal E$ of any query \Q, that we must not remove nodes from the database important for any further processing of matches.
We call this property {\em soundness of $\mathcal E$ \wrt \Q}.
\begin{definition}\label{def:soundness}
  Let \db be a graph database, \Q a \sparql query and $\mathcal E$ any SOI representation of \Q with solutions $S\subseteq \vars(\Q)\times O_\db$.
  $\mathcal E$ is {\em sound \wrt \Q} iff for the largest solution $S$ of $\mathcal E$, it holds that if $\mu(v)=o$ for some $v\in\vars(\Q)$ and $\mu\in\lbr\Q\rbr_\db$, then $(v,o)\in S$.
\end{definition}

\subsection{Advanced Graph Patterns}
\label{sub:advanced}
BGPs, and \sparql queries in general, may be combined by operators, further restricting and linking the sets of matches.
This subsection is devoted to applying dual simulation principles to queries with \sunion- and \sand-operators.
The \sand-operator is best characterized by relational inner-joins of the results of two queries.

The \sunion-operator is the least invasive operator.
It combines any two queries $\Q[1]$ and $\Q[2]$ to query $\Q[1]\sunion \Q[2]$.
The result set is the union of the result sets of the constituent queries, \ie $\lbr \Q[1]\sunion\Q[2]\rbr_{\db} := \lbr\Q[1]\rbr_\db \cup \lbr\Q[2]\rbr_\db$.
It is well-known that any \sparql query may be rewritten as the union of finitely many {\em union-free} queries\conf{ (\cf Proposition~3.8~\cite{Perez2009})}.
A \sparql query \Q is {\em union-free} if the \sunion-operator does not occur in \Q.
\begin{report}
\begin{proposition}[Proposition 3.8~\cite{Perez2009}]
  Let \Q be a \sparql query.
  Then there are union-free \sparql queries $\Q[1],\Q[2],\ldots,\Q[k]$ ($k\in\mathbb{N}$) such that \Q is equivalent to $\Q' = \Q[1]\sunion \Q[2]\sunion \ldots \sunion \Q[k]$, \ie $\lbr\Q\rbr_\db = \lbr\Q'\rbr_\db$.
\end{proposition}
The construction of $\Q'$ follows similar principles as constructing the DNF (disjunctive normal form) in propositional logic.
In consequence, the result set of \Q is the union of the result sets of all the \Q[i] ($1\leq i\leq k$).
\end{report}
Instead of \Q we may process each union-free part of \Q individually and later combine their results.
Henceforth, we assume every query to be union-free.

While \sparql's disjunction unifies the result sets of the constituents, conjunction unifies compatible results, \ie those results agreeing upon shared variables.
Matches $\mu_1$ and $\mu_2$ are {\em compatible}, denoted $\mu_1 \compat \mu_2$, if for all $v\in\dom(\mu_1)\cap\dom(\mu_2)$ ($v$ shared by $\mu_1$ and $\mu_2$), $\mu_1(v)=\mu_2(v)$.
The conjunction of two queries \Q[1] and \Q[2] is the query $\Q[1]\sand \Q[2]$.
As an example, the \sparql representation of the graph pattern in Fig.~\ref{fig:transitivity}(a) may be described as the conjunction of two BGPs, $\mathbb{G}_1 = \{ ( \mathtt{v}, \texttt{knows}, \mathtt{w} ) \}$ and $\mathbb{G}_2 = \{ ( \mathtt{w}, \texttt{knows}, \mathtt{v} ) \}$.
The semantics of conjunctions is defined by \\[.3em]
\centerline{$\lbr\Q[1]\sand\Q[2]\rbr_\db := \{ \mu_1 \cup \mu_2 \mid \mu_i \in\lbr\Q[i]\rbr_\db \wedge \mu_1 \compat \mu_2 \}$.}\\[.3em]
For example, in the database in Fig.~\ref{fig:transitivity}(b), queries $\mathbb{G}_i$ from above enjoy matches $\mu_i$ ($i=1,2$) with $\mu_1(\mathtt{v})=\mu_2(\mathtt{v})=p_1$ and $\mu_1(\mathtt{w})=\mu_2(\mathtt{w})=p_2$.
These matches are compatible, thus $(\mu_1\cup\mu_2)\in\lbr\mathbb{G}_1 \sand \mathbb{G}_2\rbr_\db$.
In contrast, $\mu_1$ from before and $\mu_3$ with $\mu_3(\texttt{w})=p_2$ and $\mu_3(\texttt{v})=p_3$ constitute incompatible matches, thus $(\mu_1 \cup \mu_3)\notin\lbr\mathbb{G}_1 \sand \mathbb{G}_2\rbr_\db$.

Regarding our dual simulation process, for conjunctions $\Q[1]\sand\Q[2]$, we create the systems of inequalities for $\Q[1]$ and $\Q[2]$ separately, denoted by $\mathcal E(\Q[1])$ and $\mathcal E(\Q[2])$.
Recall that the variables of both queries directly refer to variables occurring in \Q[1] and \Q[2], respectively.
The semantics of conjunctions requires matches to queries \Q[1] and \Q[2] to be compatible.
In consequence, assignments to common variables must be identical.
This may be achieved by simply unifying the systems of inequalities of both queries.
The following lemma defines the sound system of inequalities.
\begin{lemma}\label{ref:lemma-conjunction}
  Let \db be a graph database and $\Q[1],\Q[2]$ BGPs or conjunctions with sound systems of inequalities $\mathcal E(\Q[1]) = (\mathtt{Var}_1,\mathtt{Eq}_1)$ and $\mathcal E(\Q[2]) = (\mathtt{Var}_2,\mathtt{Eq}_2)$.
  Then $\mathcal E = (\mathtt{Var}_1\cup\mathtt{Var}_2, \mathtt{Eq}_1\cup\mathtt{Eq}_2)$ is sound for $\Q[1]\sand\Q[2]$.
\end{lemma}
\begin{report}
\begin{proof}
  Let $\mu\in\lbr \Q[1]\sand\Q[2]\rbr_\db$.
  It holds that $\mu=\mu_1\cup\mu_2$ for compatible $\mu_i\in\lbr\Q[i]\rbr_\db$ ($i=1,2$).
  Let $v\in\vars(\Q[1]\sand\Q[2])$ with $\mu(v)=o$.
  We need to show that the largest solution $S$ of $\mathcal E$ contains $(v,o)$.
  In case $v\in\vars(\Q[1])\cap\vars(\Q[2])$, it holds that $\mu_1(v)=\mu_2(v)=o$.
  Hence, the largest solutions $S_i$ of $\mathcal E(\Q[i])$ ($i=1,2$) contain $(v,o)$, \ie $(v,o)\in S_1 \cap S_2$, because $\mathcal E(\Q[i])$ are sound.
  It remains to be shown that $S_1 \cap S_2 \subseteq S$.
  Let $(v,o)\in S_1 \cap S_2$.
  By construction, any $\epsilon\in\mathtt{Eq}$ either comes from $\mathtt{Eq}_1$ or $\mathtt{Eq}_2$, and since $(v,o)\in S_1\cap S_2$, $(v,o)$ cannot contradict $\epsilon$.
  Thus, $(v,o)$ belongs to the largest solution $S$.
  
  In the other case we have $v\in\vars(\Q[i])\setminus\vars(\Q[j])$ ($i,j=1,2$ and $i\neq j$).
%   Recall that $\mu_1$ and $\mu_2$ are compatible matches, \ie for all shared variables between \Q[1] and \Q[2], the assignments agree on the values of these variables.
  Of course $(v,o)\in S_i$.
  The only way $(v,o)\notin S$ holds is if there is a shared variable $w$ that is connected to $v$, via one or more triple patterns, and every possible assignment to $w$ disagrees with $S_2$.
  However, there is at least one object $S_1$ and $S_2$ have to agree upon for $w$, because $\mu_1$ and $\mu_2$ are compatible assignments.
  Hence the assumption leads to a contradiction and $(v,o)\in S$.
%   Since $S_1$ and $S_2$ are simulations, objects $o$ and $o'$ where $(w,o')\in S_i$
%   , $(v,o)\in S$ follows from soundness of $\mathcal E(\Q[i])$, \unsure{since variable $v$ cannot be influenced by any triple pattern of $\mathcal E(\Q[j])$.}
\end{proof}
\end{report}
\begin{conference}
A proof can again be found in our technical report~\cite{Mennicke2019TR}.
\end{conference}

\subsection{Optional Patterns}
\label{sub:optional}
The last syntactic construct of \sparql for which we provide a sound dual simulation procedure is that of optional patterns.
While, in terms of complexity, it is the most involved \sparql operator~\cite{Schmidt2010FoundationsOptimization}, our procedure needs rather small adjustments.
Reconsider our introductory query \X[1], where we asked for directors and their coworkers.
If we are not sure whether every director has a person listed they worked with, then we may put this information in an optional pattern, yielding query \X[2].
\begin{center}
\scalebox{.8}{\mbox{
\begin{tabular}{l}
{\tt SELECT $^*$ WHERE \{} \\
{\tt ~~?director directed ?movie .} \\
{\tt ~~OPTIONAL \{} \\
{\tt ~~~~?director worked\_with ?coworker . \} \} }
\end{tabular}
}}
\hspace{1em}\X[2]
\end{center}
Optional patterns are left-outer joins in the relational model, \ie matches to \X[2] definitely assign nodes from the database to variable {\tt ?director} and {\tt ?movie}, but to variable {\tt ?coworker} only if there is one.
Regarding the graph database in Fig.~\ref{fig:database}(a), we obtain all bold subgraphs, as before, and additionally the semi-thick subgraphs (with {\tt D. Koepp} and {\tt T. Young} as {\tt ?director}).
In general, for queries \Q[1] and \Q[2], the result set of $\Q[1]\sand\Q[2]$ is contained in the result set of the optional pattern $\Q[1]\soptional\Q[2]$.
Additionally, all matches to \Q[1] that have no compatible matches to \Q[2] are matches, \ie\\[.3em]
%Formally,\\[.3em]
%
\centerline{\scalebox{.85}{$\begin{array}{rcl}
  \lbr\Q[1]\soptional\Q[2]\rbr_\db & := & \lbr\Q[1]\sand\Q[2]\rbr_\db \cup \\
  &&  \{ \mu\in\lbr\Q[1]\rbr_\db \mid \not\exists\mu'\in \lbr\Q[2]\rbr_\db : \mu \compat \mu' \}\text.
\end{array}$}}\\[.3em]
In \X[2], variable {\tt ?director} occurs in two different roles.
First, the optional pattern mandates variable {\tt ?director} to feature triples with label {\tt directed}.
Second, triples labeled {\tt worked\_with} are only optional.
These two roles must be reflected by our SOI representation of \X[2] by including two copies of that variable, $\texttt{?director}_m$ (mandatory) and $\texttt{?director}_o$ (optional) with the property that a solution $S$ in variable $\texttt{?director}_o$ must not exceed $S$ in variable $\texttt{?director}_m$.
In other words, there is no database node matching $\texttt{director}_o$ that does not match $\texttt{?director}_m$.
This is expressed by inequality
\begin{equation}\label{eq:optional}
  \texttt{?director}_o \leq \texttt{?director}_m\text.
\end{equation}
To faithfully describe such dependencies, we need to distinguish optional variable occurrences from mandatory ones, based the formal query syntax.

The query language $\mathcal S$ comprises union-free \sparql queries with \sand and \soptional operators, as the following grammar describes:
$$\begin{array}{rcc|c|c}
  \Q & ::= & \mathbb{G} & \Q \sand \Q & \Q \soptional \Q
\end{array}$$
where $\mathbb{G}$ ranges over by BGPs.
Queries in $\mathcal S$ range over by $\Q, \Q[1], \Q[2], \ldots$.
As observed above, we need to consider mandatory and optional variable occurrences.
Function $\mathit{mand}$ maps queries \Q from $\mathcal S$ to the set of variables that occur as mandatory in \Q, defined by
\begin{enumerate}
  \item $\mathit{mand}(\mathbb{G}):=\vars(\mathbb{G})$,
  \item $\mathit{mand}(\Q[1]\sand \Q[2]):= \mathit{mand}(\Q[1])\cup\mathit{mand}(\Q[2])$, and
  \item $\mathit{mand}(\Q[1]\soptional \Q[2]):= \mathit{mand}(\Q[1])$.
\end{enumerate}
For handling optional pattern $\Q[1]\soptional\Q[2]$ correctly, we need to decide, in which cases an occurrence of variable $\mathtt{v}$ in \Q[2] has an optional dependency to another occurrence of the same variable.
The case $\mathtt{v}\in\vars(\Q[1])$ is reflected by query \X[2].
Upon identification of such mandatory/optional pairs, we rename the optional occurrences of variables in our SOI and add an inequality as before, \eg \eqref{eq:optional}.
More precisely, for the special case of query $\Q = \Q[1] \soptional \Q[2]$, we create the SOI representation for \Q by first identifying mandatory/optional dependencies between $\Q[1]$ and $\Q[2]$, that are occurrences of variables $\mathtt{v}\in\vars(\Q[2])\cap\mathit{mand}(\Q[1])$.
For $\mathtt{v}\in\vars(\Q[2])\cap\mathit{mand}(\Q[1])$, we reserve a unique name $\mathtt{v}_{\Q[2]}$, which we use to replace $\mathtt{v}$ in every inequality of \Q[2], achieved by a renaming $\rho := \{ (\mathtt{v},\mathtt{v}_{\Q[2]}) \mid \mathtt{v}\in\vars(\Q[2])\cap\mathit{mand}(\Q[1]) \}$.
Upon renaming, we add inequality
\begin{equation}\label{eq:optional-less}
  \mathtt{v}_{\Q[2]} \leq \mathtt{v}
\end{equation}
for $\mathtt{v}\in\vars(\Q[2])\cap\mathit{mand}(\Q[1])$ to the overall SOI.
The largest solution to the resulting SOI consists of all assignments to the new variables $\mathtt{v}_{\Q[2]}$, \ie to variables not occurring in the original formulation of the query.
Since these variables are only surrogates necessary for handling optionality correctly, and the largest solution for these variables is subsumed by the respective mandatory variables (\cf \eqref{eq:optional-less}), we may ignore them in the final result of the pruning step.
\begin{lemma}\label{lemma:optional-soundness}
  Let \db be a graph database and $\Q[1],\Q[2]$ two \sparql queries with sound systems of inequalities $\mathcal E(\Q[1]) = (\mathtt{Var}_1,\mathtt{Eq}_1)$ and $\mathcal E(\Q[2]) = (\mathtt{Var}_2,\mathtt{Eq}_2)$.
  Furthermore, define renaming as $\rho$ by $\rho(\mathtt{v}):=\mathtt{v}_{\Q[2]}$ for all $\mathtt{v}\in\vars(\Q[2])\cap\mathit{mand}(\Q[1])$.
  Then \\[.3em]
  \scalebox{1}{\centerline{
  $\begin{array}{rcl}
    \mathcal E & = & (\mathtt{Var}_1 \cup \mathtt{Var}_2 \cup \rho(\mathtt{Var}_2), \mathtt{Eq}_1 \cup \rho(\mathtt{Eq}_2) \cup \mathtt{Eq}_0)
  \end{array}$
  }}\\[.3em]
  with $\mathtt{Eq}_0 := \{ \mathtt{v}_{\Q[2]} \leq \mathtt{v} \mid \mathtt{v}\in\vars(\Q[2])\cap\mathit{mand}(\Q[1]) \}$ is sound for $\Q[1]\soptional\Q[2]$.
\end{lemma}
\begin{conference}
A more detailed description of variable dependencies in optional patterns is included in the technical report~\cite{Mennicke2019TR}.
\end{conference}
\begin{report}
\begin{proof}
  Let $\mu\in\lbr\Q[1]\soptional\Q[2]\rbr_\db$ with $\mu(v)=o$ for $\mathtt{v}\in\vars(\Q[1]\soptional\Q[2])$.
  We need to show that $(\mathtt{v},o)\in S$ where $S$ is the largest solution of $\mathcal E$.
  There are two cases to distinguish, \begin{inparaenum}[(a)]
    \item $\mu = \mu_1 \cup \mu_2$ where $\mu_i\in\lbr\Q[i]\rbr_\db$ ($i=1,2$) with $\mu_1 \compat \mu_2$ and
    \item $\mu = \mu_1$ where $\mu_1\in\lbr\Q[1]\rbr_\db$ and there is no $\mu_2\in\lbr\Q[2]\rbr_\db$ compatible to $\mu_1$.
  \end{inparaenum}
  Case (a) becomes analogous to the proof of Lemma~\ref{ref:lemma-conjunction}, considering that for any occurrence of $\mathtt{v}_{\Q[2]}$ in $\rho(\mathtt{Eq}_2)$, inequality~(\ref{eq:optional-less}) makes the requirements upon $\mathtt{v}_{\Q[2]}$ only weaker.
  Hence, $\mu_2(\mathtt{v})$ is preserved.
  In case (b), we distinguish two further cases for variable $\mathtt{v}$, \begin{inparaenum}[(i)]
    \item $\mathtt{v}\in\vars(\Q[1])\setminus\vars(\Q[2])$ and
    \item $\mathtt{v}\in\vars(\Q[1])\cap\vars(\Q[2])$.
  \end{inparaenum}
  The claim for case (i) directly follows from the sound SOI $\mathcal{E}(\Q[1])$.
  In case (ii), it might be that in the largest solution $S_2$ of $\mathcal{E}(\Q[2])$, $(\mathtt{v},o)\notin S_2$.
  However, $\mathtt{v}$ is subject to renaming, since it is a variable of both sub-queries.
  Therefore $(\mathtt{v}_{\Q[2]},o)\notin \rho{S_2}$ but as we added inequality~(\ref{eq:optional-less}) to $\mathtt{Eq}_o$, we get that $(\mathtt{v},o)\in S$ by soundness of $\mathtt{E}(\Q[1])$.
\end{proof}
\end{report}
\begin{report}
\subsection{The General Case}
The general case, outlined by example \X[3], needs to take the contexts of optional patterns into account.
Since $\mathtt{v_3}$ has a mandatory occurrence in \X[3] but an optional in the sub-query $\{ (\mathtt{v_1}, a, \mathtt{v_2}) \} \soptional \{ (\mathtt{v_3}, b, \mathtt{v_2}) \}$, \sparql's evaluation semantics defines the second occurrence of $\mathtt{v_3}$ to be mandatory \wrt the first.
For any optional pattern $\Q[1]\soptional\Q[2]$ occurring as a sub-query of a query $\Q\in\mathcal S$, if a variable $\mathtt{v}\in\vars(\Q[2])$ occurs as mandatory in $\Q$, then we perform the same renaming as in Lemma~\ref{lemma:optional-soundness} for \Q[2].
For a variable $\mathtt{v}\in\vars(\Q[2])$, there may be several candidates.
From all the choices we pick the {\em syntactically closest}.
As an example, consider the optional patterns
$$\begin{array}{rcl}
  P & = & ( P_1 \soptional P_2 ) \soptional P_3\text{ and} \\
  R & = & R_1 \soptional ( R_2 \soptional R_3 )\text.
\end{array}$$
Assume that $\mathtt{y}\in\vars(P_i)$ ($i=1,2,3$) and $\mathtt{z}\in\vars(R_i)$ ($i=1,2,3$).
The occurrences of $\mathtt{y}$ in $P_2$ and $P_3$ are syntactically closest to the mandatory occurrence of $\mathtt{y}$ in $P_1$, giving rise to inequalities
$$\begin{array}{rcl}
  \mathtt{y}_{P_2} & \leq & \mathtt{y} \\
  \mathtt{y}_{P_3} & \leq & \mathtt{y}\text.
\end{array}$$
It may also be that $\mathtt{x}\in\vars(P_i)$ ($i=2,3$) and $\mathtt{x}\notin\vars(P_1)$.
In these situations, we rename $\mathtt{x}$ to $\mathtt{x}_{P_2}$ and $\mathtt{x}_{P_3}$, respectively, but would not add any interdependencies between these variables.
In extreme cases, the original variable $\mathtt{x}$ may not occur in the resulting SOI at all.
In these cases, the soundness proof requires that every solution to $\mathtt{x}_{P_2}$ or $\mathtt{x}_{P_3}$ also is a solution to variable $\mathtt{x}$.

The occurrence of $\mathtt{z}$ in $R_3$ is closest to the occurrence in $R_2$, and the occurrence in $R_2$ is closest to $R_1$, raising the following inequalities,
$$\begin{array}{rcl}
  \mathtt{z}_{R_3} & \leq & \mathtt{z}_{R_2} \\
  \mathtt{z}_{R_2} & \leq & \mathtt{z}\text.
\end{array}$$
Handling the general case formally, needs to conduct a notion of $\mathcal S$-contexts, being queries with holes.
Since the proof of the resulting soundness lemma is completely analogous, thus gives no more insights than the proof of Lemma~\ref{lemma:optional-soundness}, our considerations about optional patterns are complete.
\end{report}
What if optional patterns occur within the clauses of a conjunction?
Let us consider another example:
\begin{center}
  \scalebox{.8}{\mbox{$\left(\{ (\mathtt{v_1}, a, \mathtt{v_2}) \} \soptional \{ (\mathtt{v_3}, b, \mathtt{v_2}) \}\right) \sand \{ (\mathtt{v_3}, c, \mathtt{v_4}) \}$\text.}}\hspace{1em}\X[3]
\end{center}
\begin{figure}[tbp]
\centering
  \begin{subfigure}[b]{.4\linewidth}
  \centering
    \scalebox{.7}{\begin{tikzpicture}
      \node[entity] (v0) {1};
      \node[entity,below left=of v0] (v1) {2}
        edge[pre] node[auto]{$a$} (v0);
      \node[entity,below right=of v0] (v2) {3}
        edge[pre] node[auto,swap]{$a$} (v0);

      \node[entity,below=of v1] (v3) {4}
        edge[post] node[auto]{$b$} (v1);

      \node[entity,below=of v2] (v4) {5}
        edge[pre] node[auto]{$d$} (v2)
        edge[pre] node[auto]{$c$} (v3);
      \node[entity,below=of v3] (v5) {6}
        edge[post] node[auto]{$d$} (v3);
    \end{tikzpicture}}
    \caption{}
  \end{subfigure}
  \quad
  \begin{subfigure}[b]{.22\linewidth}
  \centering
    \scalebox{.7}{\begin{tikzpicture}
      \node[entity,label=left:$\mathsf{v_1}$] (v0) {1};
      \node[entity,below=of v0,label=left:$\mathsf{v_2}$] (v1) {2}
        edge[pre] node[auto]{$a$} (v0);

      \node[entity,below=of v1,label=left:$\mathsf{v_3}$] (v3) {4}
        edge[post] node[auto]{$b$} (v1);

      \node[entity,below=of v3,label=left:$\mathsf{v_4}$] (v4) {5}
        edge[pre] node[auto]{$c$} (v3);
    \end{tikzpicture}}
    \caption{}
  \end{subfigure}
  \quad
  \begin{subfigure}[b]{.22\linewidth}
  \centering
    \scalebox{.7}{\begin{tikzpicture}
      \node[entity,label=left:$\mathsf{v_1}$] (v0) {1};
      \node[entity,below=of v0,label=left:$\mathsf{v_2}$] (v1) {3}
        edge[pre] node[auto]{$a$} (v0);

      \node[entity,below=of v1,label=left:$\mathsf{v_3}$] (v3) {4};

      \node[entity,below=of v3,label=left:$\mathsf{v_4}$] (v4) {5}
        edge[pre] node[auto]{$c$} (v3);
    \end{tikzpicture}}
    \caption{}
  \end{subfigure}
  \caption{(a) Graph Database, (b) and (c) Matches of \X[3]}\label{fig:non-well-designed}
\end{figure}
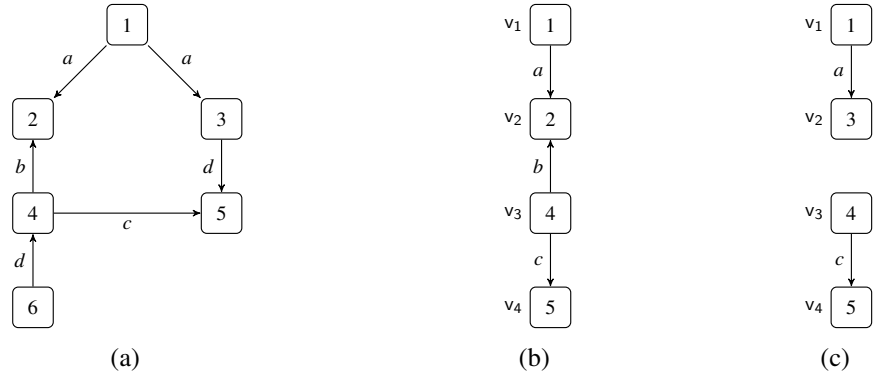
The query consists of three triple patterns, the first two constitute an optional pattern and their results are joined with the third triple pattern.
Fig.~\ref{fig:non-well-designed}(b) and (c) show possible matches of \X[3] \wrt the graph database in Fig.~\ref{fig:non-well-designed}(a).
Analogous to \X[2], we derive $\mathtt{v_2}_m$ and $\mathtt{v_2}_o$ with $\mathtt{v_2}_o \leq \mathtt{v_2}_m$ from the optional pattern.
The first occurrence of $\mathtt{v_3}$ is optional whilst the second occurrence is a mandatory one.
Matches to $\mathtt{v_3}$ have an outgoing $c$-labeled edge and may feature the $b$-labeled edge from the optional pattern.
Although both occurrences are not directly linked in an optional pattern, the second occurrence restricts the possible assignments for the first one.
Formally, we express such relations as a generalization of Lemma~\ref{ref:lemma-conjunction}.
We use renamings based on a unique identification of subqueries.
If we abstract \X[3] to $\Q[1]\sand\Q[2]$, then $\Q[1]$ is $R_1 \soptional R_2$.
In the course of renaming, we replace $\mathtt{v_3}$ in $R_2$ by a fresh variable, \eg $\mathtt{v_3}^{R_2}$, and add inequality $\mathtt{v_3}^{R_2} \leq \mathtt{v_3}$ to the system of inequalities for \X[3].
Renaming functions $\rho_i$ ($i=1,2$) are defined accordingly to rename variables that occur only optional in $\Q[i]$ but mandatory in the other subquery $\Q[j]$.
\begin{lemma}\label{lemma:general-conjunction}
  Let \db be a graph database and $\Q[1],\Q[2]\in\mathcal S$ with sound systems of inequalities $\mathcal E(\Q[1]) = (\mathtt{Var}_1,\mathtt{Eq}_1)$ and $\mathcal E(\Q[2]) = (\mathtt{Var}_2,\mathtt{Eq}_2)$.
  Define renaming $\rho_i$ ($i=1,2$) as given above, and $\mathtt{Eq}_0 := \{ \mathtt{v}' \leq \mathtt{v} \mid (\mathtt{v},\mathtt{v})\in\rho_i, i\in\{ 1,2\} \}$.
  Then $\mathcal E = (\mathtt{Var}_1\cup\mathtt{Var}_2\cup\rho_1(\mathtt{Var}_1)\cup\rho_2(\mathtt{Var}_2), \rho_1(\mathtt{Eq}_1)\cup\rho_2(\mathtt{Eq}_2)\cup\mathtt{Eq}_0)$ is sound for $\Q[1]\sand\Q[2]$.
\end{lemma}
%
% The proof is completely analogous to the one of Lemma~\ref{ref:lemma-conjunction}.
Independence of the results due to \Q[1] and \Q[2] is guaranteed by the renaming function and the additional inequalities.
The solution is interpreted as if all renamed variables are unified with their originals.

\subsection{Discussion} % (fold)
\label{sub:final-discussion}
Before we discuss an important query type, we conclude this section by showing that the constructed systems of inequalities are sound for any query $\Q\in\mathcal S$, using all the results we obtained so far in the proof.
%The proof summarizes the pruning procedure by the construction of a respective SOI.
%
\begin{theorem}[Soundness]\label{thm:soundness}
  Let \db be a graph database and $\Q\in\mathcal S$.
  Then $\mathcal E(\Q)$ is a sound SOI.
\end{theorem}
%
% The proof consists of an induction over the structure of \Q and uses all results obtained so far.
%
\begin{report}
\begin{proof}
%   We proceed by induction over the structure of \Q.
  For the base case, $\Q=\mathbb{G}$, Theorem~\ref{thm:bgp-soundness} provides us with the necessary argument.
  Since the largest dual simulation is the largest solution of the respective SOI, soundness of $\mathcal{E}(\mathbb{G})$ immediately follows.
  Assume for queries $\Q[1],\Q[2]\in\mathcal S$, soundness of the respective SOIs $\mathcal E(\Q[1])$ and $\mathcal E(\Q[2])$ is already provided, which may already conduct some renaming due to our discussion in Sect.~\ref{sub:optional}.
  %Please note that for the recursive step, a relation $S$ is sound for $\Q[i]$ ($i=1,2$) iff it is sound up to the renaming $\rho_i$ involved (\cf Lemma~\ref{lemma:optional-soundness}).
%
  For the recursive step, we distinguish two cases.
  First, if $\Q=\Q[1]\sand \Q[2]$, then $\mathcal{E}(\Q)$ is sound due to Lemma~\ref{ref:lemma-conjunction}\footnote{Adjustments to the soundness notion has no influence on the lemma's correctness.}.
  Lemma~\ref{lemma:optional-soundness} proves soundness of $\mathcal E(\Q)$ with $\Q=\Q[1]\soptional \Q[2]$.
\end{proof}
\end{report}

Our theoretical considerations are limited to \sparql queries in which every node of a triple pattern is a variable.
\sparql also allows mentioning constants, \ie objects and literals from the database, often drastically reducing the number of possible results.
The key to integrating constant nodes into our pruning technique is to alter inequality \eqref{eq:init1}.

Our dual simulation process is not restricted to well-designed patterns.
Well-designed patterns are \sparql queries \Q with the property that for every sub-query $\Q[1]\soptional \Q[2]$ and every $\mathtt{v}\in\vars(\Q[2])$ that also occurs outside the optional pattern also occurs in \Q[1], \ie $\mathtt{v}\in\vars(\Q[1])$~\cite{Perez2009}.
Query \X[3] is not well-designed, since $\mathtt{v_3}$ occurs as an optional variable but also outside the optional sub-pattern.
Non-well-designed patterns give rise to cross-product results, as indicated by the match in Fig.~\ref{fig:non-well-designed}(c).
Assume that we have several $c$-labeled edges, then each of these edges together with the $a$-labeled edge forms an answer to the query.
In these situations, our procedure remains effective, since it handles both occurrences of variable $\mathtt{v_3}$ separately.
In fact, the addition of \sand and \soptional operators does not influence the complexity of our procedure.
Considering dual simulation as a query processor for $\mathcal S$, {\scshape Pspace}-completeness of the evaluation problem~\cite{Schmidt2010FoundationsOptimization} may be evaded, since checking whether a given relation $S$ constitutes a valid assignment to $\mathcal E(Q)$ and extensions of it may be performed in {\scshape Ptime}.
More expressive fragments of \sparql add combinatorial complexity not solvable by pure dual simulation pattern matching.
% However, notions like {\em query result} and {\em compatibility} must be generalized to simulations.
% \begin{conference}
% In the next section, we discuss experiments performed with our software prototype implementing the dual simulation process.
% \end{conference}

\begin{report}
There are two reasons which make well-designed patterns interesting.
First, the fragment containing only well-designed patterns has a {\scshape coNP}-complete evaluation problem~\cite{Perez2009,Arenas2013}, as opposed to {\scshape Pspace}-completeness of \sparql's evaluation problem.
Second, every well-designed pattern is {\em weakly monotone}~\cite{Arenas2013}, an important property when discussing {\em NULL} semantics.
% A query is {\em weakly monotone} iff for any graph databases $\db, \db'$ with $\db\subseteq\db'$ (\ie \db is a subgraph of $\db'$), it holds that $\lbr\Q\rbr_\db \sqsubseteq \lbr\Q\rbr_{\db'}$~\cite{Arenas2013}.
% Intuitively, relation $\sqsubseteq$ relates two result sets $\mathcal R_1$ and $\mathcal R_2$ if every match in the first is essentially present in the second, \ie for each $\mu\in\mathcal R_1$, there is a $\mu'\in\mathcal R_2$ such that $\dom(\mu)\subseteq\dom(\mu')$ and $\mu\compat \mu'$ (\ie they are compatible).
%
% It is unknown whether well-designedness characterizes weakly monotone queries~\cite{Arenas2013} in that there might be queries that are non-well-designed but weakly monotone.
To this end, we cannot tell whether or not we handle all weakly monotone queries effectively.
However, the next section provides indications in this respect.
%An example query is derived from \X[3] as {\bfseries (\Q[4])},
%\begin{center}
%  \mbox{$\left(\{ (\mathtt{v_1}, a, \mathtt{v_2}) \} \soptional \{ (\mathtt{v_3}, b, \mathtt{v_2}) \}\right) \sand \{ (\mathtt{v_3}, c, \mathtt{v_4}) \}$}\hspace{1em}{\bfseries (\Q[4])}
%\end{center}
\end{report}

%
% \section{Query Pruning Processes}\label{sec:query-optimization}
% %
% \input{tex/query-optimization}
%
\section{Evaluation}\label{sec:evaluation}
%
% -*- root: ../main.tex -*-
%
% In this section, we evaluate our prototype called {\scshape sparqlSim}.
First, we compare our algorithm to the state-of-the-art dual simulation algorithm as introduced by Ma \etal~\cite{Ma2014} and used in implementations of \cite{Mottin2016,Xie2017PoPanda,Ma2014} for evaluation purposes.
Both are implemented within our prototype called {\scshape sparqlSim}.
Second, we analyze how our \sparql extension of dual simulation may be used to effectively and efficiently prune graph databases to improve query processing on an in-memory RDF database and a triple store based on relational database technology.
After analyzing the effectiveness of the pruning, we compare query evaluation times with two graph database systems on two very large graph datasets comprising 750 million and 1.3 billion triples.
We focus on time-consuming optional queries which were also used by Atre~\cite{Atre2015LeftDescriptors}.
Details concerning the evaluation results, a list of queries, and our implementation can be found on our project's Github page.

\subsection{Experimental Setup}
For the first experiment, we have implemented the dual simulation algorithm of Ma \etal as an option in our tool.
To evaluate our prototypes' performance as a pruning mechanism, we employed one of the fastest RDF databases Virtuoso~\cite{Erling2009RDFDBMS} and the high-performance in-memory database RDFox~\cite{Nenov2015}.
All experiments have been performed on a server running Ubuntu 16.04 with four XEON E7-8837, \SI{2.67}{\GHz}, having 8 Cores each, \SI{384}{\giga\byte} RAM and a Kingston DCP1000 NVMe PCI-E SSD.
We deactivated caching for Virtuoso to achieve stable query evaluation times.
RDFox is not using query caches.
For the evaluation, we have run all queries 10 times on each database and averaged the times.

Since we provide a dual simulation algorithm that can be used as an external pruning mechanism, we imported the result sets from our tool into the two databases manually and then re-evaluated the queries on the pruning in comparison to queries on the full databases.
Here, we did not consider the export time from our tool and the import time into the database, because our tool could easily be integrated into a standard database system, using our computations internally.

\begin{table}[t]
\centering
\caption{Runtimes of our {\scshape sparqlSim} for BGPs from queries \B[0]-\B[21] compared to Ma \etal~\cite{Ma2014}.}
\label{tab:simulation}
\begin{tabular}{lrr|lrr}
\toprule
Query & $t_{\textsc{sparqlSim}}$ & $t_{\textsc{Ma et al.}}$ & Query & $t_{\textsc{sparqlSim}}$ & $t_{\textsc{Ma et al.}}$ \\
\midrule
\B[0] & 0.10385 & 6.72121 & \B[10] &0.02397 & 0.27126 \\
\B[1] & 0.03876 & 3.33471 & \B[11] &0.01392 & 0.02099 \\
\B[2] & 0.79097 & 3.84781 & \B[12] &0.01477 & 0.02287 \\
\B[3] & 0.69797 & 5.62662 & \B[13] &0.35515 & 11.30355 \\
\B[4] & 0.00003 & 0.00004 & \B[14] &5.46599 & 16.63957 \\
\B[5] & 0.04091 & 0.31700 & \B[15] &13.43710 & 24.99660 \\
\B[6] & 0.41105 & 0.54291 & \B[16] &0.00002 & 0.00003 \\
\B[7] & 0.26991 & 0.51206 & \B[17] &1.12649 & 2.30390 \\
\B[8] & 0.13562 & 5.51084 & \B[18] &0.32056 & 0.54057 \\
\B[9] &0.02551 & 0.08707 & \B[19] &0.69515 & 5.15070 \\
\bottomrule
\end{tabular}
\end{table}

Our evaluation data comprises two popular RDF datasets:
\begin{inparaenum}[(1)]
   \item The DBpedia dump 2016-10 in the English language version~\cite{Auer07} and
   \item the synthetic Lehigh University Benchmark~\cite{Guo05} (LUBM) dataset generated for \num{10000} universities.
\end{inparaenum}
DBpedia comprises \num{751603507} triples with \num{216132665} nodes and \num{65430} predicates.
While the DBpedia queries \D[0]-\D[5] stem from \cite{Atre2015LeftDescriptors}, benchmark queries \B[0]-\B[19] appeared in the DBpedia benchmark dataset in \cite{Morsey2011}.
The LUBM benchmark dataset comprises \num{1381692508} triples with \num{18} predicates and \num{328620750} nodes.
Since official query sets hardly cover optional patterns, we rely on queries that have been used by Atre~\cite{Atre2015LeftDescriptors} (\cf \lubm[0]-\lubm[5]).

The space our tool allocates for storing the adjacency matrices sums up to \SI{35}{\giga\byte} for LUBM and \SI{23}{\giga\byte} for DBpedia.
The biggest matrices of LUBM consume between \SI{1}{\giga\byte} and \SI{4}{\giga\byte} of main memory (11 out of 36, \eg {\tt rdf:type}).
\num{99}\% of the DBpedia predicates allocate less than \SI{1}{\mega\byte}.
Constructing the adjacency matrices and producing the result triples requires additional space for storing maps and string objects.

\subsection{Evaluation Analysis}
\paragraph{Comparison of Dual Simulation Algorithms}
Due to the fact that Ma \etal's algorithm~\cite{Ma2014} considers BGPs as input we have removed the \sparql keyword \soptional from benchmark queries \B[0]-\B[19].
Evaluation times are shown in Table~\ref{tab:simulation}.
We observe that the optimizations allowed by {\scshape sparqlSim} (\cf Sect.~\ref{sub:discussion}) pay off, since we outperform Ma \etal's algorithm in every case, often even by an order of magnitude.
When running in graph database query scenarios, it is this order of magnitude the naive algorithm lacks.

\paragraph{Dual Simulation as Pruning Mechanism}
\begin{table}
\caption{Result set sizes, numbers of required triples, runtimes of {\scshape sparqlSim} in seconds and numbers of triples after pruning.% (Note that the number of single triples after pruning can be smaller than the number of results.)
}
\label{tab:pruning}
\centering
\begin{tabular}{lrrrr}
\toprule
Query & Result No. & Req. Triples & $t_{\textsc{sparqlSim}}$ & Tripl. aft. Pruning \\
\midrule
\lubm[0] & \num{10448905} & \num{3276841} & \num{106.451} & \num{10181730} \\
\lubm[1] & \num{226641} & \num{114989} & \num{8.464} & \num{25429750} \\
\lubm[2] & \num{32828280} & \num{15416012} & \num{147.335} & \num{48674046} \\
\lubm[3] & \num{11} & \num{35} & \num{0.138} & \num{126} \\
\lubm[4] & \num{10} & \num{33} & \num{0.125} & \num{101} \\
\lubm[5] & \num{7} & \num{35} & \num{1.220} & \num{35} \\
\midrule
\D[0] & \num{523066} & \num{3139273} & \num{4.396} & \num{3141102} \\
\D[1] & \num{0} & \num{0} & \num{0.002} & \num{0} \\
\D[2] & \num{12} & \num{60} & \num{0.088} & \num{60} \\
\D[3] & \num{5794} & \num{28704} & \num{0.143} & \num{28704} \\
\D[4] & \num{25102459} & \num{22630477} & \num{6.230} & \num{22691521} \\
\D[5] & \num{365693} & \num{79943} & \num{0.574} & \num{79944} \\
\midrule
\B[0] & \num{12} & \num{60} & \num{0.088} & \num{60} \\
\B[1] & \num{859751} & \num{726749} & \num{0.022} & \num{726812} \\
\B[2] & \num{913786} & \num{1587731} & \num{0.532} & \num{1588127} \\
\B[3] & \num{438542} & \num{386000} & \num{0.606} & \num{386020} \\
\B[4] & \num{0} & \num{0} & \num{0.000} & \num{0} \\
\B[5] & \num{0} & \num{0} & \num{0.033} & \num{0} \\
\B[6] & \num{815522} & \num{886826} & \num{0.503} & \num{886939} \\
\B[7] & \num{34991} &  \num{37965} & \num{0.443} & \num{37965} \\
\B[8] & \num{8416} & \num{30258} & \num{0.113} & \num{30258} \\
\B[9] & \num{8247} & \num{13116} & \num{0.022} & \num{13116} \\
\B[10] & \num{8061} & \num{12642} & \num{0.027} & \num{12642} \\
\B[11] & \num{9849} & \num{8955} & \num{0.018} & \num{8955} \\
\B[12] & \num{9554} & \num{8660} & \num{0.018} & \num{8660} \\
\B[13] & \num{123467} & \num{365131} & \num{0.273} & \num{365154} \\
\B[14] & \num{22673220} & \num{27652055} & \num{4.322} & \num{27747192} \\
\B[15] & \num{0} & \num{0} & \num{0.000} & \num{0} \\
\B[16] & \num{2} & \num{4} & \num{0.009} & \num{4} \\
\B[17] & \num{7898331} & \num{8285964} & \num{0.917} & \num{8294385} \\
\B[18] & \num{66903} & \num{41808} & \num{0.472} & \num{41808} \\
\B[19] & \num{879460} & \num{292531} & \num{0.602} & \num{292541} \\
\bottomrule
\end{tabular}
\end{table}
First, we analyze {\scshape sparqlSim}'s pruning effectiveness (\cf Table~\ref{tab:pruning}) of dual simulation for all LUBM and DBpedia queries.
Observe that the number of triples is drastically decreased from the original databases for all queries.
\begin{report}
For queries with 0 triples left, there is no need for any further query evaluation.
\end{report}
Over all tested queries we prune at least 95\% of the original database.
Hence, for most DBpedia queries we prune all triples not required for any result (compare req. triples and tripl. aft. pruning in Table~\ref{tab:pruning}).
In comparison, the effectiveness of our pruning is smaller for LUBM queries, being least effective for query \lubm[1].
\begin{report}
Here, only \num{0.9}\% of the triples after pruning are actually part of some result.
\end{report}
Later on we provide evidence that, \eg for \lubm[1], our pruning allows the two database systems to enormously improve upon their evaluation times.

Regarding efficiency, {\scshape sparqlSim}'s evaluation time heavily depends on the query and the dataset.
With LUBM, having only 18 distinct predicates, we have an extreme case that often needs over 30 iterations to compute the largest dual simulation, which leads to high running times of our algorithm, \eg for \lubm[0] or \lubm[2].
As an outstanding characteristic, these two queries have a huge number of results.
It is further a combination of the cyclic shape of the queries and the low selectivity of the predicates within the queries that explains the long runtime of our algorithm.
In DBpedia, predicates usually have a much higher selectivity.
Hence, we usually perform the computation for these queries in only a split-second.

\paragraph{Runtime of RDF Databases}
\begin{table}[t]
%\footnotesize
\centering
   \caption{Query processing times on the full and pruned dataset, and query times including pruning times for RDFox. All times are measured in seconds.}
  \label{tab:rdfox}
  \begin{tabular}{lrrr}
  \toprule
    Query &  $t_{\text{DB}}$ & $t_{\text{DB pruned}}$ & $t_{\text{DB pruned}} + t_{\textsc{sparqlSim}}$\\
    \midrule
    \lubm[0] & 19.100 & 1.401 & 107.852 \\
    \lubm[1] & \num{25900.000} & 888.000 & 896.464 \\
    \lubm[2] & 161.000 & 15.690 & 163.025 \\
    \lubm[3] & 0.000 & 0.000 & 0.138 \\
    \lubm[4] & 0.000 & 0.000 & 0.125 \\
    \lubm[5] & 0.000 & 0.000 & 1.223 \\
    \midrule
    \D[0] & 1.400 & 1.115 & 5.511 \\
    \D[1] & 0.000 & 0.000 & 0.002 \\
    \D[2] & 1.100 & 0.003 & 0.091 \\
    \D[3] & 0.620 & 0.002 & 0.145 \\
    \D[4] & 5.960 & 3.493 & 9.722 \\
    \D[5] & 3.230 & 0.016 & 0.590 \\
    \midrule
    \B[0]  & 1.468 & 0.000 & 0.088 \\
    \B[1]  & 0.099 & 0.030 & 0.052 \\
    \B[2]  & 0.348 & 0.110 & 0.642 \\
    \B[3]  & 0.104 & 0.012 & 0.618 \\
    \B[4]  & 0.033 & 0.000 & 0.000 \\
    \B[5]  & 0.000 & 0.000 & 0.033 \\
    \B[6]  & 12.830 & 0.042 & 0.545 \\
    \B[7]  & 14.410 & 0.002 & 0.445 \\
    \B[8]  & 0.793 & 0.001 & 0.114 \\
    \B[9]  & 0.117 & 0.001 & 0.023 \\
    \B[10] & 0.004 & 0.001 & 0.028 \\
    \B[11] & 0.001 & 0.000 & 0.018 \\
    \B[12] & 0.001 & 0.001 & 0.019 \\
    \B[13] & 0.643 & 0.022 & 0.295 \\
    \B[14] & 3.282 & 1.998 & 6.320 \\
    \B[15] & 0.941 & 0.000 & 0.000 \\
    \B[16] & 0.000 & 0.000 & 0.009 \\
    \B[17] & 0.758 & 0.310 & 1.227 \\
    \B[18] & 0.119 & 0.001 & 0.473 \\
    \B[19] & 18.750 & 0.048 & 0.650 \\
    \bottomrule
  \end{tabular}
\end{table}
By the next experiments we compare the query evaluation time of the in-memory database RDFox to {\scshape sparqlSim} in combination with RDFox as a query processor.
In Table~\ref{tab:rdfox}, we observe an improvement of the query time in 15 out of 32 queries.
Especially the improvement on query \lubm[1] with a query processing time of \num{25900} seconds on RDFox is worth mentioning.
We could run our dual simulation algorithm in only 8 seconds (\cf Table~\ref{tab:pruning}), decreasing the query time of RDFox by more than 20 times.
For \lubm[0], however, $t_{\textsc{sparqlSim}}$ alone is around 5 times slower than RDFox ($t_{\text{DB}}$).
Also, in queries \D[5], \B[0], \B[7]-\B[9], \B[17], \B[21] we show good improvements of the in-memory databases' evaluation times.
For most of the remaining queries we show comparable results to RDFox, varying by some milliseconds.

\begin{table}[t]
\centering
   \caption{Query processing times on the full and pruned dataset, and query times including pruning times for Virtuoso. All times are measured in seconds.}
  \label{tab:virtuoso}
  \begin{tabular}{lrrr}
  \toprule
    Query &  $t_{\text{DB}}$ & $t_{\text{DB pruned}}$ & $t_{\text{DB pruned}} + t_{\textsc{sparqlSim}}$\\
    \midrule
    \lubm[0] & 5.126 & 2.261 & 108.712 \\
    \lubm[1] & 50.853 & 0.971 & 9.435 \\
    \lubm[2] & 56.676 & 26.767 & 174.102 \\
    \lubm[3] & 0.001 & 0.000 & 0.138 \\
    \lubm[4] & 0.000 & 0.000 & 0.125 \\
    \lubm[5] & 0.000 & 0.000 & 1.223 \\
    \midrule
    \D[0] & 0.395 & 0.359 & 4.755 \\
    \D[1] & 0.001 & 0.000 & 0.002 \\
    \D[2] & 0.002 & 0.000 & 0.089 \\
    \D[3] & 0.010 & 0.003 & 0.147 \\
    \D[4] & 2.148 & 4.008 & 10.238 \\
    \D[5] & 0.039 & 0.021 & 0.595 \\
    \midrule
    \B[0] & 0.002 & 0.000 & 0.088 \\
    \B[1] & 0.003 & 0.001 & 0.023 \\
    \B[2] & 0.003 & 0.003 & 0.030 \\
    \B[3] & 0.001 & 0.002 & 0.020 \\
    \B[4] & 0.001 & 0.002 & 0.020 \\
    \B[5] & 0.054 & 0.031 & 0.303 \\
    \B[6] & 1.082 & 0.441 & 4.762 \\
    \B[7] & 0.000 & 0.000 & 0.000 \\
    \B[8] & 0.000 & 0.000 & 0.009 \\
    \B[9] & 0.121 & 0.099 & 1.016 \\
    \B[10] & 0.043 & 0.009 & 0.031 \\
    \B[11] & 0.012 & 0.003 & 0.476 \\
    \B[12] & 0.102 & 0.056 & 0.658 \\
    \B[13] & 0.069 & 0.064 & 0.596 \\
    \B[14] & 0.000 & 0.000 & 0.000 \\
    \B[15] & 0.000 & 0.000 & 0.034 \\
    \B[16] & 0.042 & 0.026 & 0.594 \\
    \B[17] & 0.022 & 0.013 & 0.516 \\
    \B[18] & 0.003 & 0.001 & 0.444 \\
    \B[19] & 0.021 & 0.005 & 0.118 \\
    \bottomrule
  \end{tabular}
\end{table}

Table~\ref{tab:virtuoso} shows an improvement of the running times of only 3 queries for Virtuoso.
For most other queries, evaluation times are on par with $t_{\text{DB}}$.
For some queries, our pruning could not increase Virtuoso's evaluation time as much as for RDFox.
A detailed analysis of Virtuoso's query plans revealed that this was due to changes in the join order that sometimes seems to turn against optimal evaluation times by drastically increasing the number of intermediate results, \eg \D[4] with doubled evaluation time $t_{\text{DB pruned}}$ on the 3\% portion of DBpedia.
Nevertheless, we believe that Virtuoso could benefit from a direct integration of {\scshape sparqlSim} as a pruning technique.
In turn, our tool may advance by employing Virtuoso's built-in heuristics for query evaluation plans.
On the downside, our algorithm is often slightly slower than the professionally implemented and highly optimized RDF triple store.
Some of the more complex queries took longer to produce the pruning than for Virtuoso to produce the actual answers.
These queries took several iterations in {\scshape sparqlSim} to compute.
We believe that we can benefit from more sophisticated join order optimization techniques as used for example in Virtuoso which could boost our computation times tremendously.
The very fast pruning time for the cyclic query \lubm[1] requires only two iterations, and thereby points to the potential of our solution.

\subsection{Discussion}
\label{sub:eval-discussion}
The evaluation results suggest dual simulation pruning as an effective technique allowing two state-of-the-art graph database systems to improve upon their query evaluation times, sometimes enormously.
Preprocessing \lubm[1] is most profitable, since huge intermediate tables can be avoided.
In this case we observe a decrease by more than one order of magnitude while the pruning time is vastly fast in only two iterations.
In contrast, because intermediate results in the evaluation of \lubm[0] are rather small, the benefits of dual simulation pruning are not as significant as for \lubm[1].
Furthermore, the low selectivity predicates of \lubm[0] result in a rather big number of iterations that increases the pruning time compared to \eg \lubm[1].
As a general rule we recommend using dual simulation for pruning in cases where queries produce large intermediate results.
Such cases can usually be detected employing database statistics for join result size estimation, also used for join order optimization.
\begin{conference}
Our technical report~\cite{Mennicke2019TR} contains more details of \lubm[0] and \lubm[1].
Remarkably, the complexity of both queries may not be found in the optional patterns but in their cyclic shape.
\end{conference}

\begin{report}
The queries discussed so far are outstanding in their own roles.
While \lubm[0]'s evaluation is always faster than the computation of the dual simulation pruning, both database systems we considered benefit from the pruning for \lubm[1].
The mandatory cores of both queries are depicted in Fig.~\ref{fig:l0l1}.
First observe that both queries are cyclic.
Although \lubm[0] is quite small, our dual simulation algorithm takes more than \num{30} iterations until it reaches the fixpoint.
From a brute force analysis we learn that the number of iterations may be reduced by \num{16}, but only resulting in half the time of the computation reported in Table~\ref{tab:pruning}.
After having stabilized the equations for any two nodes of \lubm[0], the third node may turns equations for the other two nodes unstable.
Hence, the evaluation performance of Virtuoso and RDFox cannot be beaten by our current implementation, no matter which specific heuristic we choose.
The predicates of \lubm[0] share quite a low selectivity rate.
In contrast, dual simulation between query \lubm[1] and the LUBM dataset takes only two iterations, allowing for an overall improvement of Virtuoso as well as RDFox. 

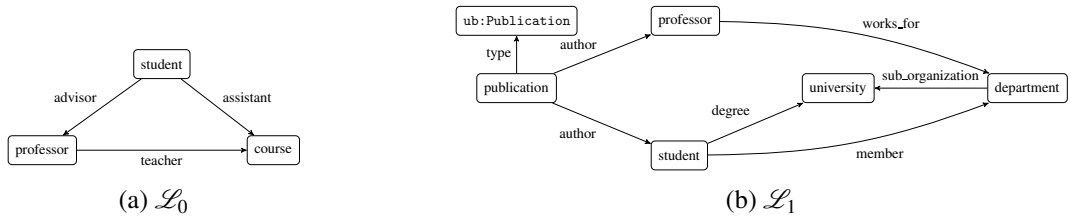
\begin{figure}
\centering
\begin{subfigure}[b]{.35\textwidth}
\centering
\scalebox{.5}{
\begin{tikzpicture}[node distance=1.5 and 1.5]
\node (h) {};
\node[entity] (stud) {student};
\node[entity,below left=of stud] (prof) {professor}
  edge[pre] node[auto]{advisor} (stud);
\node[entity,below right=of stud] (course) {course}
  edge[pre] node[auto,swap]{assistant} (stud)
  edge[pre] node[auto]{teacher} (prof);
\end{tikzpicture}
}
\caption{$\mathcal{L}_0$}
\end{subfigure}
\quad
\begin{subfigure}[b]{.6\textwidth}
\centering
\scalebox{.5}{
\begin{tikzpicture}[node distance=1 and 2.5]
\node[entity] (pub) {publication};
\node[entity,above=of pub] {\tt ub:Publication}
  edge[pre] node[auto,swap] {type} (pub);
\node[entity,below right=of pub] (student) {student}
  edge[pre] node[auto] {author} (pub);
\node[entity,above right=of pub] (professor) {professor}
  edge[pre] node[auto,swap] {author} (pub);
\node[entity,above right=of student] (univ) {university}
  edge[pre] node[auto,swap]{degree} (student);
\node[entity,right=3 of univ] (dept) {department}
  edge[post] node[auto,swap]{sub\_organization} (univ)
  edge[pre,bend left=10] node[auto]{member} (student)
  edge[pre,bend right=10] node[auto,swap]{works\_for} (professor);
\end{tikzpicture}
}
\caption{$\mathcal{L}_1$}
\end{subfigure}
\caption{The mandatory (basic graph pattern) cores of queries $\mathcal L_0$ and $\mathcal L_1$}\label{fig:l0l1}
\end{figure}

Regarding the effectiveness of the pruning, LUBM query $\lubm[1]$ represents one of the worst examples with over \num{200} times more leftover triples than necessary.
The reason for such a huge difference can be found in the counterexample to Theorem~\ref{thm:bgp-soundness} described at the end of Sect.~\ref{sub:bgp}.
Let us transfer the known example by considering a subexpression of query \lubm[1] which is depicted in Fig.~\ref{fig:l0l1}(b).
At its core, $\lubm[1]$ asks for all publications together with two of their authors, both affiliated with a department (one is a student member, the other is an employee) that is part of the university from which the student got their degree.
Suppose we have two disjoint matches isomorphic to the graph representation of $\lubm[1]$, \ie two different papers with authors from two distinct departments.
It is important that the departments belong to different universities.
Now assume the second paper has a third author who got his degree from the second university but is a student member of the first department.
Furthermore, this student has no other incident edges.
Then this student node is not part of any match due to \sparql.
However, dual simulation does not discriminate this node, since it reflects a {\em similar situation} and all adjacent nodes dual simulate their respective counterparts in $\lubm[1]$.

The LUBM dataset is especially prone to queries like \lubm[1], since it is a very large dataset with only little diversity in the generated subgraphs (recall that 18 predicates are distributed over 1.33 billion edges).
As a consequence of the low diversity, potential matches are often adjacent and dual simulation combines them frequently by edges not belonging to any match.
Custom-tailored notions of {\em query matches} based on dual simulations may avoid these (possibly) unwanted results.
However, no such solution can resolve this issue completely without stepping into \np-completeness or even \pspace-completeness.
\end{report}

\section{Related Work}\label{sec:related-work}
%
% -*- root: ../main.tex -*-
%
Recently, graph pattern matching has become a trending topic for graph databases, different from the canonical though costly prime candidate of graph isomorphism, with the goal of reducing structural requirements of the answer graphs.
Especially, simulations have been implemented for different graph database tasks~\cite{Brynielsson2010,Fan2010Simulation,Fan2012,Mottin2016}.
Ma \etal~\cite{Ma2014} introduce the notion of {\em dual simulation}.
Having a simulation preorder in a database context considering forward and backward edges is mentioned as early as in the year 2000~\cite{DataOnTheWeb}.
On the downside, performance improvements by dual simulation come with a loss of topology~\cite{Ma2014}.

Mottin \etal~\cite{Mottin2016} build on simulation as one solution to their query paradigm called {\em Exemplar Queries}. 
For a given exemplar graph pattern, the user obtains subgraphs from the database similar to the exemplar. 
We foresee that exemplar queries as well as other applications of graph pattern matching may exhibit the portion of \sparql integrated in our framework, making their proposals even more attractive to users.

Using simulation for graph database pruning has been proposed as a component in Panda~\cite{Xie2017PoPanda}.
In Panda, subgraph simulation is used to filter unnecessary tuples before answering isomorphism queries.
Their large-scale evaluation shows improvements in query time compared to several other isomorphism-based query processors.
In contrast, we rely on dual simulation being more effective in pruning unnecessary triples, and we implement a fast dual simulation algorithm operating on bit-matrices which are particularly useful for large graph databases.
Furthermore, we use a more expressive query model that could also be integrated into their pruning technique to support more complex queries.
Other existing approaches for optimizing graph database querying rely on adapting traditional database optimization techniques, usually leading to major improvements with regard to the query performance~\cite{Bornea2013BuildingDatabase,Erling2009RDFDBMS}.
However, graph database queries usually consist of numerous joins with oftentimes huge intermediate results, requiring specialized optimization techniques.
Therefore, join order estimation for graph databases, especially RDF triple stores, is still an active field~\cite{Schmidt2010FoundationsOptimization,NeumannScalableGraphs,Letelier2013,Atre2015LeftDescriptors}.
Our proposal appreciates the graph data model and performs light-weight algorithms to support traditional database optimization.
\begin{report}
In fact, upon effectiveness of {\scshape sparqlSim} huge intermediate results may be avoided.
\end{report}

Simulation-based indexing techniques have already been used for join-ahead pruning in databases on XML data~\cite{Milo1999IndexExpressions}.
The index is created by computing bisimulation equivalence classes of nodes on the original database.
Each equivalence class groups structurally bisimilar nodes~\cite{Picalausa2012ATriples,Zou2011GStore:Matching}.
Bisimulation is more restrictive than dual simulation which we use throughout this paper.
However, our algorithm could benefit from similar ideas.
It would be sufficient to produce dual simulation equivalence classes, which promises to obtain a much smaller database fingerprint than possible with bisimulations, since (dual) simulation equivalence is coarser than bisimulation.

\section{Conclusion}\label{sec:conclusion}
We proposed efficient processing of \sparql queries based on graph pattern matching.
Our algorithm builds upon dual simulation and for all extensions, due to \sparql, we provided soundness proofs.
To derive an algorithm competing with state-of-the-art graph databases we contribute an alternative characterization of dual simulation in terms of a system of inequalities.
Dual simulation is directly applicable to \sparql's BGPs, whereas composite queries including \sand and \soptional operators, are handled by conservative extensions of dual simulation.

We could show that our algorithm outperforms standard dual simulation algorithms on a variety of real-world \sparql BGPs.
Furthermore, our dual simulation algorithm can be used to aggressively prune triples, speeding up graph database query processing for state-of-the-art graph databases.
In comparison to these graph databases, we could improve the query evaluation time for several queries drastically and showed comparable results for the others.
We believe that most database systems would benefit from a direct integration of our proposal into their query processor.
Further applications already using dual simulation may benefit from our \sparql extension to offer more expressive query capabilities.

We plan to extend our prototype by applying more heuristics with which we conduct extensive experiments to find better guidelines for the applicability of dual simulation pruning.
Our experiments with two state-of-the-art graph database systems showed that such guidelines make sense on a per-system and per-data basis. 
We are currently investigating the limits of our dual simulation procedure \wrt different \sparql fragments.
While this work suggests a tremendous enhancement of the complexity of optional pattern evaluation, other operators add combinatorial problems unavoidable for a dual simulation evaluation semantics for \sparql.

\bibliographystyle{eptcs}
\bibliography{icde,local}

% \appendix
% \section{Queries}
% \input{tex/queries.tex}

\end{document}